\documentclass[10pt]{article}
\usepackage{amsmath}
\usepackage{amsfonts}
\usepackage{amssymb}
\usepackage{amscd}
\usepackage{amsthm}
\usepackage{latexsym}
\usepackage{mathrsfs}
\usepackage[all]{xy}
\usepackage{fullpage}

\newcommand{\sid}{1_\sigma}
\newcommand{\pid}{1_\pi}
\newcommand{\pidc}{\overline{1}_\pi}
\newcommand{\para}{\|}

\newcommand{\M}{\mathscr{M}}
\newcommand{\V}{\mathscr{V}}
\newcommand{\Su}{\mathscr{S}}
\newcommand{\N}{\mathscr{N}}
\newcommand{\T}{\mathscr{T}}
\newcommand{\U}{\mathscr{U}}

\newtheorem{theorem}{Theorem}[section]
\newtheorem{proposition}[theorem]{Proposition}
\newtheorem{lemma}[theorem]{Lemma}
\newtheorem{corollary}[theorem]{Corollary}

\theoremstyle{definition}

\begin{document}

\title{Taming Multirelations}

\author{Hitoshi Furusawa\\Kagoshima University,
    Japan\and Georg Struth\\University of Sheffield, UK}

\maketitle 

\begin{abstract}
  Binary multirelations generalise binary relations by associating
  elements of a set to its subsets. We study the structure and algebra
  of multirelations under the operations of union, intersection,
  sequential and parallel composition, as well as finite and infinite
  iteration. Starting from a set-theoretic investigation, we propose
  axiom systems for multirelations in contexts ranging from bi-monoids
  to bi-quantales.
\end{abstract}

\pagestyle{plain}


\section{Introduction}\label{S:introduction}

Multirelations generalise binary relations in that elements of a given
set are not related to single elements of that set, but to its
subsets. Hence a multirelation over a set $X$ is a relation of type
$X\times 2^X$ and elements of a multirelation are pairs $(a,B)$ with
$a\in X$ and $B\subseteq X$. Multirelations have found applications in
program semantics and models or logics for games that require
alternation~\cite{ChandraKozenStockmeyer}, that is, two dual kinds of
nondeterministic choice. The first one occurs within individual pairs
$(a,B)$ when selecting an output $b\in B$ to an input $a$; the second
one occurs between pairs $(a,B)$ and $(a,C)$ when selecting an output
$B$ or $C$ to an input $a$.

Multirelations occur, in slightly generalised form, as transition
relations of alternating
automata~\cite{BrzozowskiL80,ChandraKozenStockmeyer} with universal
choices $(a,\sigma,B)$ forcing simultaneous moves from state $a$ along
$\sigma$-labelled edges into all states in $B$, whereas existential
choices between $(a,\sigma, B)$ and $(a,\sigma,C)$ for instance, allow
the selection of a transition from $a$ into either $B$ or $C$ labelled
with $\sigma$.  Similarly, in two-player games, moves $(a,B)$ are made
by an antagonist, and a protagonist must be prepared to play from each
state in $B$; whereas a protagonist can select between moves $(a,B)$
and
$(a,C)$~\cite{BenthemGL08,ChandraKozenStockmeyer,Parikh83,Parikh85,PaulyParikh03}. Otherwise,
in multirelational semantics of computing
systems~\cite{Back98,CavalcantiWD06,MartinCR07,Rewitzky03,RewitzkyB06}
choices $(a,B)$ are controlled by the environment (the system must be
correct under all of them), whereas choices between $(a,B)$ and
$(a,C)$ are controlled by the system (the system must be correct under
at least one of them). Moreover, in the semantics of concurrent
dynamic
logic~\cite{Goldblatt92,NerodeWijesekera90,Peleg87,Peleg87a,WijesekeraNerode05},
an element $(a,B)$ indicates an execution of a parallel program or
component starting in state $a$ and terminating in the states within
$B$, whereas $(a,B)$ and $(a,C)$ correspond to different
executions. In these contexts, the first kind of choice is often
called \emph{external}, \emph{universal} or \emph{demonic}; the second
one is known as \emph{internal}, \emph{existential} or \emph{angelic}.

To reason about systems modelled by multirelations, modal logics such
as game logics~\cite{Parikh83,Parikh85} and concurrent dynamics
logics~\cite{Goldblatt92,NerodeWijesekera90,Peleg87,Peleg87a,WijesekeraNerode05}
have been developed. Despite its obvious significance, however, the
algebra of multirelations itself has so far received rather limited
attention---except for the special case of up-closed
multirelations~\cite{FurusawaNT09,MartinCR07,NishizawaTF09,Rewitzky03,RewitzkyB06},
which are isomorphic to monotone predicate
transformers~\cite{Parikh83}. The general algebra of multirelations,
which forms the semantics of concurrent dynamics logics and captures
alternation in its purest form, remains to be investigated in depth.

A first step towards taming multirelations has been an algebraic
reconstruction of Peleg's concurrent dynamic logic from a minimalist
axiomatic basis~\cite{FurusawaS14,FurusawaS15}.  This article takes the
next steps towards axiom systems for multirelations in the spirit of
Tarski's relation algebra (cf.~\cite{Maddux}). Though there are many
similarities, this is not straightforward: the sequential composition
of multirelations, for instance, is more complex than its relational
counterpart and not associative; the parallel composition of
multirelations has no relational counterpart; the relational converse
makes no sense for multirelations. Thus alternative axioms,
definitions and proofs for multirelational concepts are required.

Our main contributions are as follows.  To obtain a basis for our
axiom systems, we extend our previous investigation of multirelational
properties by new interaction laws between sequential and parallel
composition and an explicit definition of a multirelational domain
operation. We also identify several subclasses of
multirelations---sequential and parallel subidentities, vectors or
terminal multirelations, nonterminal multirelations---and mappings
between these classes (cf. Figure~\ref{F:magictriangle} and
\ref{F:extendeddiagram}).  Based on these set-theoretic preliminaries
we axiomatise multirelations at four different layers.

First we expand weak bi-monoids with operations for sequential and
parallel composition by sequentiality-parallelism interaction axioms.
We define a domain operation explicitly on these c-monoids and show
that the domain elements form a sub-semilattice in which sequential
and parallel composition coincide. We show that previous domain axioms
for sequential monoids \cite{DesharnaisJS09} and concurrent
monoids \cite{FurusawaS14} are derivable in this setting.

Second, we expand c-monoids to c-trioids by adding an operation of
internal choice and providing another interaction axiom. The
domain elements now form a distributive sublattice. All concurrent dynamic
algebra axioms~\cite{Peleg87,FurusawaS14} become derivable in the
presence of a Kleene star and a few more simple multirelational
properties.

Third, we study bounded distributive lattices with operations of
sequential and parallel composition, to which we add different
sequentiality-parallelism interaction axioms. The resulting c-lattices
are also c-trioids, and a large number of laws can be derived in this
setting. In particular, we prove that the algebras of subidentities
and vectors are isomorphic (only sequential composition is usually not
preserved) and characterise the subalgebras of these elements as well
as that of nonterminal elements.  In the latter, in fact, we find
greater similarity to binary relations.

Finally, we consider notions of finite and infinite iteration over
multirelations in an expansion of c-lattices to c-quantales. Due to
the lack of distributivity and associativity laws in algebras of
multirelations, our results are weaker than those for relations.

Because of the complexity of reasoning with multirelations and the
task of minimising algebraic axiom systems, we have formalised all
structures and proofs with the Isabelle/HOL theorem
prover~\cite{NipkowPW02}.  Our investigations are therefore also a
study in formalised mathematics. Our Isabelle theories and a detailed
proof document can be found online in the \emph{Archive of Formal
  Proofs}~\cite{Struth15} and consulted together with our
article. Many proofs in this article are syntactic manipulations,
which may be tedious, but carry litte insight---consequently they are
not displayed. Instead we provide human-readable Isabelle proofs
whenever suitable and we have added pointers to the facts in this
article to the Isabelle proof document. However we show some
interesting proofs and counterexamples, as they are given, but not
internally verified, by Isabelle.

The remainder of this article is organised as
follows. Section~\ref{S:preliminaries} provides the basic definitions,
operations and properties of multirelations. Section~\ref{S:seqsubids}
introduces the sets of sequential subidentities, parallel
subidentities and vectors and the isomorphisms between them. These are
summarised in the three diagrams of Figure~\ref{F:magictriangle} and
\ref{F:extendeddiagram}, on which much of the algebraic investigation
is based. It also introduces some basic properties of nonterminal
multirelations, which do not allow any pairs $(a,\emptyset)$ and
provides the multirelational laws needed for algebraic soundness
proofs in later sections. Section~\ref{S:cmonoids} studies c-monoids,
our first axiom system for multirelations, Section~\ref{S:ctrioids}
and \ref{S:clattices} introduce c-trioids and
c-lattices. Section~\ref{S:cda} explains the relationship between
c-trioids and concurrent dynamic algebras. Definitions of a domain
operation in c-lattices are presented in
Section~\ref{S:clatticedomain}.  The subalgebras of subidentities and
vectors and the associated isomorphisms are studied in
Section~\ref{S:clatticesubalgs} and \ref{S:clatticeiso}.  Functions
separating terminal and nonterminal elements are defined in
Section~\ref{S:nabladelta} and properties of these functions are
presented.  Notions of finite and infinite iterations for c-quantales
are studied in Section~\ref{S:finiteiteration} and
\ref{S:divergence}. Section~\ref{S:counterexamples} presents some
counterexamples; Section~\ref{S:upclosed} sketches how up-closed
multirelations arise in our setting. Finally,
Section~\ref{S:conclusion} presents a conclusion.


\section{Basic Definitions and Properties}\label{S:preliminaries}

This section follows Peleg~\cite{Peleg87} in defining the operations of
sequential and parallel composition on multirelations. Sequential
composition is different from the one used for up-closed
multirelations introduced by Parikh~\cite{Parikh83,Parikh85}. We
then outline a few set-theoretic properties of multirelations which
are important for the algebraic developments in later sections, and we
sketch some examples.

A \emph{multirelation} $R$ over a set $X$ is a subset of $X\times
2^X$.  We write $\M(X)$ for the set of all multirelations over $X$. The
\emph{sequential composition} of $R,S\in \M(X)$ is the multirelation
\begin{equation*}
  R\cdot S = \{(a,A)\mid \exists B.\ (a,B)\in R \wedge \exists f.\ (\forall b\in B.\ (b,f(b))\in S) \wedge A=\bigcup f(B)\}.
\end{equation*}
The \emph{unit of sequential composition} is the multirelation
\begin{equation*}
\sid = \{(x,\{x\})\mid x\in X\}.
\end{equation*}
The \emph{parallel composition} of $R$ and $S$ is the multirelation
\begin{equation*}
  R\para S = \{(a,A\cup B) \mid  (a,A)\in R\wedge (a,B)\in S\}.
\end{equation*}
The \emph{unit of parallel composition} is the multirelation 
\begin{equation*}
\pid=
\{(x,\emptyset)\mid x\in X\}.
\end{equation*} 
The \emph{universal multirelation} over $X$ is
 \begin{equation*}
U =\{(a,A)\mid a \in X \wedge A \subseteq X\}.
\end{equation*} 
 
A pair $(a,B)$ is in $R\cdot S$ if $R$ relates $a$ to some
intermediate set $C$ and $S$ relates each element $c\in C$ to a set
$B_c$ in such a way that $B=\bigcup_{c\in C} B_c$. This can be
motivated in various ways, as discussed
by~\cite{FurusawaS14}. Peleg's original intended interpretation, for
instance, associates $R\cdot S$ with the sequential composition of
parallel programs or concurrent components, such that program $R$
reaches the concurrent state in $C$ from $a$ in a parallel execution,
after which program $S$ reaches the concurrent state $B$ in parallel
executions from each state $c\in C$, reaching concurrent sub-states
$B_c$ from each of these states. Alternatively, in the context of
computation trees based on alternation or dual nondeterminism,
component $R$ makes an internal choice of moving to state $C$, which
resolves the external choices for transitions from state $a$ and
represents one level of the computation tree. After that, component
$S$ makes internal choices of moving from each state $c\in C$ to a set
$B_c$, thus resolving its external choices. The states in
$\bigcup_{c\in C} B_c$ form the next level of the computation
tree. This interpretation reflects, for instance, the construction of
a run or computation tree of an alternating automaton whose
transitions are modelled by a multirelation.

A pair $(a,B)$ is in $R\para S$ if $B$ can be decomposed with respect
to sets $C$ and $D$ such that $(a,C)\in R$ and $(a,D)\in S$, that is,
the parallel execution of $R$ and $S$ from $a$ produces the global
state $B$. In other words, the external choices in $R\para S$ arise as
the union of the external choices in $R$ and those in $S$ from each
particular state; whereas the internal choices are not combined. This
is dual to $R\cup S$, where the union of the internal choices of $R$
and $S$ is taken while the external choices are not combined.

A multirelation $R$ is a \emph{sequential subidentity} if $R\subseteq
\sid$.  The sequential subidentities form a boolean algebra with
least element $\emptyset$ and greatest element $\sid$. Join is $\cup$ and
meet is $\cdot$. The complement of a subidentity $R$ is formed by the
set $\{(a,\{a\})\mid (a,\{a\})\not\in P\}$. A \emph{parallel subidentity} is
a multirelation $R\subseteq \pid$. We write
\begin{equation*}
  \Su(X)=\{R\in \M(X) \mid R \subseteq \sid\},\qquad
 \T(X)=\{R\in \M(X) \mid R \subseteq \pid\},
\end{equation*}
for the set of all sequential and parallel subidentities of $X$. The
name $\T(X)$ is justified by the fact that parallel subidentities
can be identified with terminal multirelations; see Section~\ref{S:seqsubids}.

We consider three more sets of multirelations. A multirelation $R$ is
a \emph{vector} if whenever $(a,A) \in R$ for some $A\subseteq X$,
then $(a,A)\in A$ for all $A\subseteq X$. A multirelation $R$ is
\emph{up-closed} if $(a,A) \in R$ implies $(a,B)\in R$ for all
$B\supseteq A$ and $B\subseteq X$. Finally, we define $\pidc$ to
be the complement of $\pid$ in $\M(X)$. We write
\begin{align*}
  \V(X) &= \{R\in \M(X)\ \mid (\exists A\subseteq X. (a,A) \in
  R)\Rightarrow (\forall A \subseteq X. (a,A)\in R\},\\
\U(X) &= \{R\in \M(X)\mid \forall A\subseteq X,B\subseteq X. ((a,A) \in
R\wedge A \subseteq B\Rightarrow  (a,B) \in R\},\\
\N(X) &= \{R\in \M(X)\mid R \subseteq \pidc\}.
\end{align*}
The elements in $\N(X)$ are called \emph{nonterminal} multirelations;
see again Section~\ref{S:seqsubids}.

Multirelations form proto-dioids~\cite{FurusawaS14}, which are
defined in Section~\ref{S:ctrioids}.  At this stage it suffices to
mention the following laws.
  \begin{gather*}
R\cup (S\cup T)=(R\cup S)\cup T,\qquad
R\cup S=S\cup R,\qquad
R\cup \emptyset =R,\qquad
R\cup R=R,\\
(R\cdot S)\cdot T\subseteq R\cdot (S\cdot T),\qquad
    \sid \cdot R = R,\qquad
    R\cdot \sid = R,\\
R\cdot S \cup R\cdot T \subseteq R\cdot (S\cup T),\qquad
(R\cup S)\cdot T = R\cdot T\cup S\cdot T,\qquad
\emptyset\cdot R = R,\\
R\para (S\para T) = (R\para S)\para T,\qquad
R\para  S = S\para  R,\qquad
\pid\para  R =R,\\
R\para (S\cup T)=R\cdot S\cup R\cdot T,\qquad
\emptyset\para  R =\emptyset,\\
(R\para  S) \cdot T\subseteq (R\cdot T) \para  (S\cdot T). 
\end{gather*}
Sequential composition is not associative and $R\cdot \emptyset$ is
generally not $\emptyset$.  More generally, pairs $(a,\emptyset)\in R$
persist in any sequential composition $R\cdot S$---whence the name
\emph{terminal}.

Sequential subidentities satisfy stronger properties. First of all
$(R\cdot S)\cdot T= R\cdot (S\cdot T)$ if one of $R$, $S$, $T$ is in
$\Su(X)$.  Second,  $R\cdot (S\cup T)= R\cdot S\cup R\cdot T$ if
$R\in \Su(X)$.

The interaction between sequential and parallel composition is
captured by the following properties, among others.
\begin{lemma}\label{P:idemlemma}
Let $R\in \Su(X)$ and $S \in \T(X)$. Then
  \begin{enumerate}
  \item $R \para  R = R$, 
   \item $S \para S = S$, 
 \item $U\para  U = U$,
\item $\pidc\para  \pidc = \pidc$.
  \end{enumerate}
\end{lemma}
The proofs follow immediately from the definitions.
\begin{lemma}\label{P:interactionlemma}
Let $R,S,T\in \M(X)$. Then
  \begin{enumerate}
\item $(R\cdot \pid)\para  R = R$,
  \item $T\para  T\subseteq T \Rightarrow  (R\para  S)\cdot T= (R\cdot T) \para 
    (S\cdot T)$,
 \item $(R\para  S)\cdot T= (R\cdot T) \para 
    (S\cdot T)$, if $T\in \Su(X)\cup \T(X)\cup\{U,\pidc\}$,
\item $R\cdot (S\para  T)\subseteq (R\cdot S) \para  (R\cdot T)$,
\item $R\cdot (S\para  T) = (R\cdot T) \para  (R\cdot T)$, if $R\in
  \Su(X)$,
\item $R\cdot (S\cdot T) = (R\cdot S)\cdot T$, if one of $R$, $S$, $T$
  is in $\Su(X)\cup \T(X)$,
\item $(R\cap S)\cdot T = R\cdot T\cap S\cdot T$, if $R,S\in\Su(X)$.
  \end{enumerate}
\end{lemma}
We present one example proof to illustrate the style of reasoning with
multirelations.
\begin{proof}{(Lemma~\ref{P:interactionlemma}(2))}
  $T\para T\subseteq T$ implies that 
  \begin{equation*}
    (a,B)\in T\wedge (a,C)\in T\Rightarrow (a,B\cup C)\in T
  \end{equation*}
holds for all $a\in X$. Therefore
  \begin{align*}
    &(a,A)\in (R\cdot T)\para (S\cdot T)\\
 &\Leftrightarrow \exists B,C.\ A=B\cup C\wedge (\exists D.\ (a,D)\in
 R\wedge \exists f.\ (\forall d\in D.\ (d,f(d))\in T)\wedge B=\bigcup f(D))\\
&\hspace{3.4cm} \wedge (\exists E.\ (a,E)\in S\wedge \exists g.\ (\forall e\in E.\ (e,g(e))\in T)\wedge C=\bigcup g(E))\\
&\Leftrightarrow \exists D,E.\ (a,D\cup E)\in R\para S\wedge
\exists f,g.\ (\forall d\in D,e\in E.\ (d,f(d))\in T \wedge  (e,g(e))\in T)\\
&\hspace{5.6cm} \wedge A = \bigcup f(D)\cup g(E)\\
&\Rightarrow \exists D,E.\ (a,D\cup E)\in R\para S\wedge \exists f,g.\ (\forall d\in D-E,e\in E-D,x\in 
D\cap E.\\
&\hspace{1.5cm}(d,f(d))\in T  \wedge (e,g(e))\in T\wedge (x,(f\cup
g)(x))\in T)\\
&\hspace{1.5cm}\wedge 
 A = \bigcup f(D-E)\cup g(E-D)\cup (f\cup g)(D\cup E)\\
&\Leftrightarrow \exists D,E.\ (a,D\cup E)\in R\para S 
\wedge \exists h.\ (\forall b\in D\cup E.\ (b,h(b))\in T) \wedge A = \bigcup h(D\cup E)\\
&\Leftrightarrow \exists B.\ (a,B)\in R\para S 
\wedge \exists h.\ (\forall b\in B.\ (b,h(b))\in T) 
\wedge A = \bigcup h(B)\\
&\Leftrightarrow (a,A)\in (R\para S)\cdot T. 
  \end{align*}
  The third step uses the assumption $T\para T\subseteq T$;  the fourth
  one uses the function 
  \begin{equation*}
    h(x) =
    \begin{cases}
      f(x) &\text{if $x\in D-E$},\\
(f\cup g)(x) &\text{if $x\in D\cap E$},\\
g(x) &\text{if $x\in E-D$}. 
    \end{cases}
  \end{equation*}
The converse inclusion has been proved in~\cite{FurusawaS14}.
\end{proof}


\section{Subalgebras and Isomorphisms}\label{S:seqsubids}

This section studies the relationship between the sets of sequential
subidentities, parallel subidentities and vectors as well as the
special case of nonterminal multirelations.  We use these sets and
their relationships to extract algebraic axioms for multirelations in
later sections.  Most of the properties outlined in this section are
verified rigorously in the algebraic setting of later sections.

The constants $\sid$, $\pid$, $\emptyset$, $U$ and $\pidc$ play an important role in our
considerations.  The first lemma of this section describes their action on
multirelations.

\begin{lemma}\label{P:constprops}
  Let $R\in \M(X)$. Then 
  \begin{enumerate}
 \item $R\cdot \pid =\{(a,\emptyset) \mid \exists B.\  (a,B)\in R\}$,
  \item $R\cap\pid = R\cdot\emptyset =\{(a,\emptyset) \mid
    (a,\emptyset) \in R\}$,
  \item $R\cap\sid = \{(a,\{a\}\mid (a,\{a\}) \in R\}$,
\item $R\cdot U = \{(a,A)\mid A=\emptyset \wedge (a,A)\in R\}\cup
  \{(a,A)\mid \exists B\neq\emptyset.\ (a,B)\in R\}$,
  \item $R\para  U = \{(a,A)\mid \exists B.\ (a,B) \in R \wedge B\subseteq 
  A\}$.
\end{enumerate}
\end{lemma}
The proofs are straightforward. Intuitively, $R\cdot \pid$ replaces
every pair $(a,A)\in R$ by $(a,\emptyset)$, overwriting $A$ by
$\emptyset$, whereas $R\cap \pid$ and $R\cdot\emptyset$ both project
on the pairs $(a,\emptyset)\in R$. The multirelation $R\cap \sid$
projects on the pairs $(a,\{a\})\in R$; and the product $R\para U$
computes the up-closure of $R$. Thus
\begin{align*}
\Su(X) &= \{R\in\M(X)\mid R\cap \sid = R\},\\
\T(X) &= \{R\in\M(X)\mid R\cdot\pid = R\},\\
  \V(X)&=\{R\in \M(X)\mid (R\cdot\pid)\para U= R\},\\
  \U(X)&=\{R\in\M(X)\mid R\para U = R\},\\
  \N(X)&=\{R\in\M(X)\mid R\cap \pidc = R\}.
\end{align*}
The functions $(\cap\sid)=\lambda x.\ x\cap\sid$, $(\cdot\pid)=
\lambda x.\ x\cdot\pid$, $(\para U)=\lambda x.\ x\para U$, whose
fixpoints determine the sets $\Su(X)$, $\T(x)$ and $\V(X)$, and the
map $(\cdot U)=\lambda x.\ x\cdot U$ play an important structural
role. When specialising their sources and targets to $\Su(X)$, $\V(X)$
or $\T(X)$ they serve as bijective pairs showing that these sets are
isomorphic. More precisely, the maps in Figure~\ref{F:magictriangle}
are isomorphisms, and we verify this fact in Proposition~\ref{P:iso2}.
\begin{figure}[h]
  \centering 
\begin{equation*}
  \def\labelstyle{\normalsize}
\xymatrix @R=6pc @C=6pc{
 \Su(X)\ar@/^/[d]^{(\cdot U)}\ar@/^/[r]^{(\cdot \pid)} &
 \T(X)\ar@/^/[l]^{(\para  \sid)}\ar@/^/[dl]^{(\para  U)}\\
 \V(X)\ar@/^/[ur]^{(\cdot \pid)}\ar@/^/[u]^{(\cap \sid)}
}
\end{equation*}
  \caption{Isomorphisms between $\Su(X)$, $\T(X)$ and $\V(X)$}
\label{F:magictriangle}
\end{figure}
Under the source and target restrictions indicated, each function in
the diagram is bijective and pairs of functions between the same sets
compose to identity maps of the appropriate type.  The isomorphism
between $\Su(X)$ and $\V(X)$ is well known in the setting of binary
relations; the same maps are used for implementing it. The
isomorphisms with $\T(X)$ are particular to multirelations. More
generally it can be shown that all triangles in this diagram commute.

The pair $(\para \sid)\circ (\cdot\pid) = 1_{\Su(X)}$ generalises to a
map $\lambda x.\ (x\cdot\pid)\para \sid:\M(X)\to \Su(X)$. Applied to a
multirelation $R$, the function $(\cdot \pid)$ overwrites any $A$ in a
pair $(a,A)\in R$ by $\emptyset$; the function $(\para\sid)$ further
overwrites it by $\{a\}$. In other words,
\begin{equation*}
  (R\cdot\pid)\para\sid =\{(a,\{a\})\mid \exists B. (a,B)\in R\},
\end{equation*}
which represents the domain of $R$. We can thus define the domain of a
multirelation  explicitly as
\begin{equation*}
  d(R)=(R\cdot\pid)\para\sid,
\end{equation*}
that is, the following diagram commutes.
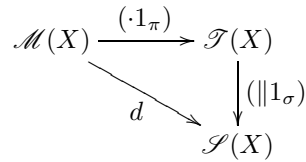
\begin{equation*}
  \def\labelstyle{\normalsize}
\xymatrix @R=2pc @C=3pc{
 \M(X)\ar[dr]_{d}\ar[r]^{(\cdot \pid)} &
 \T(X)\ar[d]^{(\para  \sid)}\\
& \Su(X) 
}
\end{equation*}
Moreover, $d(R)=R\cap \sid$ if $R\in\V(X)$ and $d(R)=R\para\sid$ if
$R\in\T(X)$, such that the maps $(\cap\sid)$ and $(\para\sid)$ in
Figure~\ref{F:magictriangle} can be replaced by $d$. The result is
shown in the left-hand diagram of Figure~\ref{F:extendeddiagram}.
\begin{figure}[h]
  \centering 
\begin{equation*}
  \def\labelstyle{\normalsize}
\xymatrix @R=2pc @C=2pc{
\M(X)\ar[dr]^{d}\ar@/^1pc/[drrr]^{(\cdot \pid)}&&&\\
 &\Su(X)\ar@/^/[dd]^{(\cdot U)}\ar@/^/[rr]^{(\cdot \pid)} &&
 \T(X)\ar@/^/[ll]^{d}\ar@/^/[ddll]^{(\para  U)}\\
&&&\\
 &\V(X)\ar@/^/[uurr]^{(\cdot \pid)}\ar@/^/[uu]^{d}
}
\xymatrix @R=2pc @C=2pc{
\N(X)\ar[dr]^{d}\ar@/^1pc/[drrr]^{(\cdot \pid)}\ar@/_1pc/[dddr]_{(\cdot \pidc)}&&&\\
 &\Su(\N(X))\ar@/^/[dd]^{(\cdot \pidc)}\ar@/^/[rr]^{(\cdot \pid)} &&
 \T(\N(X))\ar@/^/[ll]^{d}\ar@/^/[ddll]^{(\para  \pidc)}\\\
&&&\\
 &\V(\N(X))\ar@/^/[uurr]^{(\cdot\pid)}\ar@/^/[uu]^{d}
}
\end{equation*}
  \caption{Isomorphisms between $\Su(X)$, $\T(X)$ and $\V(X)$ for
    general and nonterminal multirelations}
\label{F:extendeddiagram}
\end{figure}
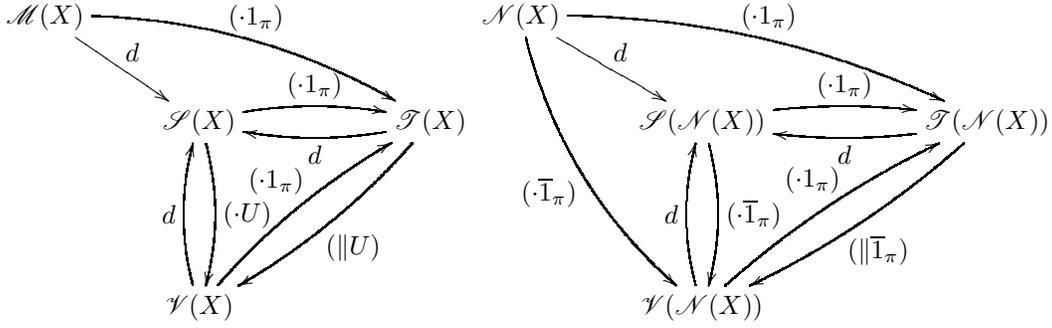
The right-hand diagram shows the situation restricted to nonterminal
multirelations. The sets of sequential subidentities,
parallel subidentities and vectors remain isomorphic, but in the
bijections, $\pidc$ replaces $U$.  In addition, vectors are now
obtained by $(\cdot \pidc):\M(X)\to\V(X)$ similarly to binary
relations.

The following properties can be justified from these diagrams; they are
needed for verifying soundness of the algebraic axioms in later
sections.
\begin{lemma}\label{P:mraxioms}
  \begin{enumerate}
  \item $((R \cdot \pid) \para \sid)\cdot S = (R\cdot\pid)\para S$,
\item $R\cdot\pid \subseteq \pid$,
\item $R\cdot\pid\cup R\cdot\pidc = R\cdot U$,
\item $\pid \cap (R \cup \pidc) = R\cdot \emptyset$,
\item $((R \cap \sid)\cdot \pid)\para \sid = R\cap \sid$,
\item $((R\cap\pidc)\cdot \pid)\para \sid = \sid \cap (R\cap 
  \pidc)\cdot \pidc$,
\item $((R\cap\pidc)\cdot\pid)\para \pidc =(R\cap \pidc)\cdot\pidc$. 
  \end{enumerate}
\end{lemma}

Equation (1) states that $d(R)\cdot S= (R\cdot\pid)\para S$,
generalising commutation of the diagrams $(\cdot U)\circ d= (\para
U)\circ (\cdot\pid)$ and $(\cdot \pidc)\circ d =(\para\pidc)\circ
(\cdot \pid)$ in Figure~\ref{F:extendeddiagram} to $(\cdot S)\circ d =
(\para S)\circ (\cdot \pid)$. For $S=\sid$, (1) specialises to the
explicit domain definition $d(R)=(R\cdot\pid)\para \sid$. Equation (2)
reflects the fact that $R\cdot\pid\in\V(X)$. Equations (3) and (4)
arise from the fact that $\pidc$ and $\pid$ are complements. Equation
(5) states that $d(R)=R$ for all $R\in\Su(X)$, which is the
isomorphism condition
$(\para\sid)\circ(\cdot\pid)=1_{\Su(X)}$. Equation (6) states that
$d(R)=\sid\cap R\cdot\pidc$ for all $R\in\N(X)$, that is, the diagram
$d=d\circ (\cdot\pidc) =(\cap\sid)\circ (\cdot\pid)$ commutes. This
domain definition is analogous to the relational case; however it does
not generalise to $\M(X)$. Equation (7) states that
$d(R)\cdot\pidc=R\cdot\pidc$ for all $R\in\N(X)$, that is, the diagram
$(\cdot\pidc)=(\cdot\pidc)\circ d$ commutes. Again, this reflects a
relational fact, which does not generalise to $\M(X)$.

The domain of a binary relation can be defined explicitly as well:
either as $d(R)=1\cap R\cdot U$ or as $d(R)=1\cap R\cdot R^\smile$,
where $1$ denotes the identity relation, $U$ the universal relation
and $R^\smile$ the converse of $R$. However, for
$R=\{(a,\emptyset)\}$, we have $\sid\cap R\cdot U=\emptyset\subset
\{(a,\{a\})\}=d(R)$ and the converse of a multirelation does not seem
to make sense.


\section{c-Monoids}\label{S:cmonoids}

We now introduce a first axiom system for multirelations in a
minimalist bi-monoidal setting, where only sequential and parallel
composition and the corresponding units are present. It allows us to
use the explicit domain definition
\begin{equation}
  d(x) = (x\cdot \pid)\para  \sid,\label{eq:domdef}\tag{d}
\end{equation}
which has been verified in the multirelational model, to derive domain
axioms similar to those for domain monoids~\cite{DesharnaisJS09}, and
to verify some properties of the $\M(X)$, $\Su(X)$, $\T(X)$ triangle.

A \emph{proto-monoid} is a structure $(S,\cdot,\sid)$ such that $1_
\sigma \cdot x = x = x\cdot \sid$ holds for all $x\in S$. Hence
composition is not required to be associative. A
\emph{proto-bi-monoid} is a structure $(S,\cdot,\para ,\sid,\pid)$
such that $(S,\cdot,\sid)$ is a proto-monoid and $(S,\para ,\pid)$ a
commutative monoid. A \emph{concurrent monoid} (\emph{c-monoid}) is a
proto-bi-monoid that satisfies the axioms
\begin{align}
 (x\cdot \pid) \para  x &=x, \label{eq:c1}\tag{c1}\\
  ((x\cdot \pid)\para  \sid)\cdot y &= (x \cdot \pid)\para 
  y, \label{eq:c2}\tag{c2}\\
(x\para  y)\cdot \pid &= (x\cdot \pid)\para (y\cdot \pid), \label{eq:c3}\tag{c3}\\
(x\cdot y)\cdot \pid &= x\cdot (y\cdot \pid),\label{eq:c4}\tag{c4}\\
\sid\para \sid &=\sid .\label{eq:c5}\tag{c5}
\end{align}
These axioms are sound with respect to the multirelational model.
 \begin{proposition}\label{P:dommonoid}
  $(\M(X),\cdot,\para ,\sid,\pid)$ forms a c-monoid. 
\end{proposition}
 \begin{proof}
   The following axioms have been verified in the following Lemmas:
   (\ref{eq:c1}) in \ref{P:interactionlemma}(1); (\ref{eq:c2}) in
   \ref{P:mraxioms}(1); (\ref{eq:c3}) in
   \ref{P:interactionlemma}(2) and \ref{P:idemlemma}(2); (\ref{eq:c4})
   in \ref{P:interactionlemma}(6); (\ref{eq:c5})
   in \ref{P:idemlemma}(1).
  \end{proof}
  Next we prove a property that is crucial for showing closure of
  subalgebras.
  \begin{lemma}\label{P:cmonoidretractions}
    In every c-monoid, the maps $d$ and $(\cdot\pid)=\lambda x.\ x\cdot\pid$ are
    retractions, that is, $d(d(x))= d(x)$ and $(x\cdot\pid)\cdot\pid =
    x \cdot\pid$.
  \end{lemma}
  It is a general property of a retraction $f:X\to X$ that $x\in
  f(X)\Leftrightarrow f(x)=x$, where $f(X)$ is the image of $S$ under
  $f$.  For every c-monoid $S$ we define the sets of domain elements
  and terminal elements
  \begin{equation*}
    d(S)=\{x\in S\mid d(x)=x\},\qquad \T(S)=\{x\in S\mid x\cdot\pid = x\}
  \end{equation*}
  and use the fixpoint characterisations $d(x)=x$ and $x\cdot\pid = x$
  for typing their elements. Sets of sequential subidentities, vectors
  and nonterminal elements, however, cannot yet be expressed in the
  c-monoid setting.

\begin{lemma}\label{P:auxdomprops}
   In every c-monoid,
  \begin{enumerate}
  \item $d(x)\cdot y = (x\cdot \pid)\para  y$,
\item $d(x\cdot \pid)\cdot y = (x\cdot \pid)\para  y$,
\item  $d(x)\cdot \pid = x\cdot \pid$,
\item $d(x\cdot\pid)= d(x)$,
\item $\pid \cdot \pid= \pid$,
\item $d(\pid) = \sid$.
  \end{enumerate}
\end{lemma}
Properties (1) and (2) paraphrase axiom (\ref{eq:c2}). In particular,
(2) implies that $d(w)=w\para \sid$ for each $w\in\T(S)$. These facts
are shown in the left-hand subdiagram of
Figure~\ref{F:extendeddiagram} below. Properties (3) and (4)
correspond to the right-hand subdiagram.

\begin{equation*}
  \def\labelstyle{\normalsize}
\xymatrix @R=2pc @C=2pc{
S\ar[r]^d\ar[d]_{(\cdot\pid)} &
d(S)\ar[d]^{(\cdot y)}\\
\V(S)\ar[r]_{(\para y)}& S 
}
\qquad\qquad 
\xymatrix @R=2pc @C=2.5pc{
S\ar[r]^{\cdot\pid}\ar[dr]_{d} &
\T(S)\ar@/^/[d]^{d}\\
& d(S)\ar@/^/[u]^{\cdot\pid}
}
\end{equation*}

We can now show that the sets $d(S)$ and $\T(S)$ are isomorphic.
\begin{proposition}\label{P:dtermiso}
  Let $S$ be a c-monoid. The two functions
     $d:\T(S)\to d(S)$ and $(\cdot\pid):d(S)\to \T(S)$
 form a bijective pair; the following diagrams commute.
\begin{equation*}
  \def\labelstyle{\normalsize}
\xymatrix @R=2pc @C=.2pc{
&\T(S)\ar[dr]^d\\
d(S)\ar[ur]^{(\cdot\pid)}\ar[rr]_{1_{d(S)}}&&d(S) 
}
\qquad\qquad 
\xymatrix @R=2pc @C=.2pc{
&d(S)\ar[dr]^{(\cdot\pid)}\\
\T(S)\ar[ur]^{d}\ar[rr]_{1_{\V(S)}}&&\T(S) 
}
\end{equation*}
Therefore
$  d(S)\cong \T(S)$.
\end{proposition}
\begin{proof}
  Functions are bijections if and only if they are invertible. It
  therefore suffices to check that $d(d(x)\cdot \pid)=d(x)$ and
  $d(x\cdot\pid)\cdot\pid = x\cdot \pid$, which follows directly from
  Lemma~\ref{P:auxdomprops}(3) and (4) or from commutation of the
  associated diagram.
\end{proof}

Next we derive the domain axioms proposed in~\cite{DesharnaisJS09}.
\begin{proposition}\label{P:domaxverif}
In every c-monoid,
  \begin{enumerate}
  \item $d(x\para  y)= d(x)\para  d(y)$,
  \item $d(x)\para  d(y)= d(x)\cdot d(y)$,
  \item $d(x)\cdot x = x$,
  \item $d(x\cdot d(y))= d(x\cdot y)$,
  \item $d(d(x)\cdot y) = d(x)\cdot d(y)$,
  \item $d(x)\cdot d(y)=d(y)\cdot d(x)$,
  \item $d(\sid)= \sid$.
  \end{enumerate}
\end{proposition}
All axioms (\ref{eq:c1})-(\ref{eq:c5}) are needed in these
proofs. Isabelle/HOL generates counterexamples in their absence.
Equations (1) and (2) are domain proto-trioid
axioms~\cite{FurusawaS14}, which are part of the axiomatisation of
concurrent dynamic algebras; the others are domain monoid
axioms~\cite{DesharnaisJS09}.

We can now characterise the subalgebra of domain elements, whereas the
c-monoid axioms are too weak to characterise that of parallel
subidentities.
  \begin{proposition}\label{P:domsemilat}
    Let $S$ be a c-monoid. Then $d(S)$ forms a sub-semilattice with
    multiplicative unit $\sid$ in which sequential and parallel
    composition coincide.
  \end{proposition}
\begin{proof}
  The closure conditions are verified by checking $d(d(x)\cdot
  d(y))=d(x)\cdot d(y)$, $d(d(x)\para d(y))=d(x)\para d(y)$ and
  $d(\sid)=\sid$. Sequential and parallel composition of
  subidentities coincides due to Lemma~\ref{P:domaxverif}(2). It
  remains to check that sequential composition of domain elements is
  associative, commutative and idempotent, and that $\sid$ is a
  multiplicative unit.
\end{proof}
Depending on the order-dual interpretations of sequential composition
as join or meet, $\sid$ becomes the least or greatest semilattice
element. As similar result has been established
in~\cite{DesharnaisJS09}, but the proof does not transfer due to the
lack of associativity of sequential composition. Next we derive three
interaction laws.
\begin{lemma}\label{eq:trianglecmon}
  Let $S$ be a c-monoid. For all $x,y\in S$, $w\in d(S)$ and $z\in\T(S)$,
  \begin{enumerate}
 \item $z\para z = z$,
\item $w\para w= w$,
\item $w\cdot (x\para  y) = (w\cdot x) \para  (w\cdot y)$.
  \end{enumerate}
\end{lemma}
Finally, we refute additional interaction and associativity laws.
\begin{lemma}\label{P:domrefs}
  There are c-monoids in which, for some elements $x$, $y$ and $z$,
  \begin{enumerate}
  \item $(x\para  y)\cdot d(z) \neq (x \cdot d(z)) \para  (y\cdot d(z))$,
  \item $(x\cdot y)\cdot d(z)\neq x\cdot (y\cdot d(z))$,
\item $\pid \cdot x \neq \pid$.
  \end{enumerate}
\end{lemma}
Isabelle's counterexample generator Nitpick presents
counterexamples~\cite{Struth15}. Since they are not very revealing, we
do not show them.


\section{c-Trioids}\label{S:ctrioids}

This section considers variants of dioids, that is, additively
idempotent semirings, endowed with a sequential and a parallel
composition and with additional interaction laws between these
operations.  Our intention is to derive the domain axioms of
concurrent dynamic logic~\cite{Peleg87} from the explicit domain
definition (d) in a minimalist setting. We cannot expect to prove more
properties from Figure~\ref{F:extendeddiagram}, since neither a
maximal element $U$ nor a meet operation is available.

A \emph{proto-dioid}~\cite{FurusawaS14} is a structure $(S,+,\cdot,
,0,\sid)$ such that $(S,+,0)$ is a semilattice with least element
$0$, $(S,\cdot,\sid)$ is a proto-monoid, and the following axioms
hold:
\begin{equation*}
  x \cdot y + x \cdot z \le x\cdot (y+z),\qquad (x+y)\cdot z = x \cdot 
  z+y\cdot z,\qquad 0\cdot x = 0. 
\end{equation*}
The relation $\le$ is the standard semilattice order  defined as 
$x\le y \Leftrightarrow x+y=y$. A \emph{dioid} is a proto-dioid in 
which multiplication is associative and the left distributivity law 
$x\cdot (y+z)= x\cdot y+ x\cdot z$ holds. A dioid is 
\emph{commutative} if multiplication is. A \emph{proto-trioid} is a 
structure $(S,+,\cdot,\para ,0,\sid,\pid)$ such that 
$(S,+,\cdot,0,\sid)$ is a proto-dioid and $(S,+,\para ,0,\pid)$ a 
commutative dioid.  In every proto-dioid, multiplication is 
left-isotone, that is, $x\le y\Rightarrow z\cdot x\le z\cdot 
y$. Moreover, every proto-trioid is a proto-bi-monoid. 

A \emph{concurrent trioid} (\emph{c-trioid}) is a proto-trioid in
which the c-monoid axioms and
\begin{equation}
  x\cdot \pid \le \pid\label{eq:c6}\tag{c6}
  \end{equation}
  hold. Independence of the axioms has been checked with
  Isabelle/HOL's automated theorem provers and counterexample
  generators.

\begin{figure}[h]
  \centering 
\begin{equation*}
  \def\labelstyle{\normalsize}
\xymatrix @R=1pc @C=1pc{
&& \mathsf{cT}\\
& \mathsf{pT}\ar@{-}[ur] && \mathsf{cM}\ar@{-}[ul]\\
\mathsf{pD}\ar@{-}[ur] & &\mathsf{pbM}\ar@{-}[ul]\ar@{-}[ur]\\
& \mathsf{pM}\ar@{-}[ur]\ar@{-}[ul]
}
\end{equation*}
  \caption{Class inclusions for proto-monoids ($\mathsf{pM}$),
  proto-bi-monoids ($\mathsf{pbM}$), c-monoids ($\mathsf{cM}$),
  proto-dioids ($\mathsf{pD}$), proto-trioids ($\mathsf{pT}$) and 
  c-trioids ($\mathsf{cT}$)}
\label{F:ctrioidclasses}
\end{figure}
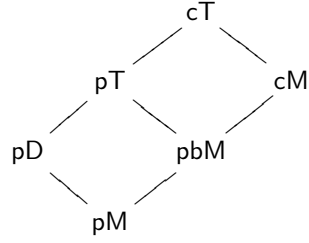

 \begin{proposition}\label{P:domtrioid}
    $(\M(X),\cup,\cdot,\para ,\emptyset,\sid,\pid)$ forms a c-trioid. 
  \end{proposition}
  \begin{proof}
    The proto-trioid axioms have been verified in~\cite{FurusawaS14};
    the c-monoid axioms in Proposition~\ref{P:dommonoid} and Axiom~(\ref{eq:c6}) in
    Lemma~\ref{P:mraxioms}(2).
  \end{proof}

In the setting of a c-trioid $S$ we can now define
\begin{equation*}
  \Su(S)=\{x\in S\mid x\le \sid\},\qquad \N(S)=\{x\in S\mid x\cdot 0 = 0\}.
\end{equation*}
It turns out that $d(S)\subseteq \Su(S)$, whereas the converse
inclusion need not hold. For $S=\emptyset$ and trioid operations
defined by $0<\pid\le \sid$ and $\pid\cdot \pid =\pid$ (remember that
$\sid \para \sid = \sid$ is an axiom) we have $\pid\le \sid$, but
$d(\pid)=(\pid\cdot \pid)\para \sid= \pid\para \sid =\sid \neq \pid$.

Similarly, $\T(S)\subseteq \{x\in S\mid x\le \pid\}$, whereas the
converse inclusion need not hold. For $S=\emptyset$ and trioid
operations defined by $0<\sid<\pid$ and $\pid\cdot\pid=\pid$ we have
$\sid\le \pid$, but $\sid\cdot\pid =\pid\neq \sid$.

This shows that, for c-trioids, the relationships between
subidentities, domain elements and fixpoints of $(\cdot\pid)$ are not
as tight as expected. The set $\N(S)$ of terminal elements is studied
in more detail in later sections. However we can derive the domain
proto-trioid axioms of~\cite{FurusawaS14}.
\begin{proposition}\label{P:cprototrioid}
  Every c-trioid satisfies the domain axioms of domain proto-trioids.
\end{proposition}
\begin{proof}
  Every c-trioid is a c-monoid, hence the properties from
  Section~\ref{S:cmonoids} hold. The following domain axioms of domain
  proto-dioids must be verified:
\begin{itemize}
 \item $x\le
  d(x)\cdot x$, which holds by Proposition~\ref{P:domaxverif}(3);
\item  $d(x\cdot
  d(y))=d(x\cdot y)$, which holds by Proposition~\ref{P:domaxverif}(4);
\item $d(x\para  y)=d(x)\para d(y)$, which holds by Proposition~\ref{P:domaxverif}(1);
\item $d(x)\para d(y)=d(x)\cdot d(y)$, which  holds by~\ref{P:domaxverif}(2);
\item $d(x+y)=d(x)+d(y)$, which holds by right distributivity of
  $\cdot$ and $\para$;
\item  $d(x)\le \sid$, which follows from (\ref{eq:c6});
\item $d(0)=0$, which is immediate from the domain definition.
\end{itemize}
\end{proof}
This does not mean, however, that every c-trioid is a domain
proto-trioid: additional axioms are assumed in the latter class. The
relationship between c-trioids and concurrent dynamic algebra is
discussed in detail in Section~\ref{S:cda}.
\begin{proposition}\label{P:ctrioidbdl}
  Let $S$ be a c-trioid. Then $d(S)$ forms a bounded distributive
  lattice with least element $0$ and greatest element $\sid$, and in
  which sequential and parallel composition coincide.
\end{proposition}
\begin{proof}
  First, by Proposition~\ref{P:domsemilat}, $d(S)$ forms a meet
  semilattice, and it is clear that $\sid$ is the greatest element with
  respect to the semilattice order.

  Second, the closure condition for $0$ has already been checked in
  the proof of Proposition~\ref{P:cprototrioid} and $d(x+y)=x+y$ holds
  for all $x,y\in d(S)$ by domain additivity and idempotency. Thus
  $d(S)$ is a join semilattice with respect to $+$ and $0$ is the
  least element with respect to the semilattice order.

  Third, for all $x,y,z\in d(S)$, the absorption laws and
  distributivity laws $x+ x\cdot y=x$, $x\cdot (x+y)=x$, $(x+y)\cdot
  z=x\cdot z+y\cdot z$ and $x+y\cdot z=(x+y)\cdot (x+z)$ must be
  verified.
\end{proof}
Once more, this result is similar to that for domain 
proto-trioids~\cite{FurusawaS14}, but proofs need to be revised due to 
different axioms.

Finally sequential composition of domain elements and parallel
composition of parallel subidentities are greatest lower bound
operations, hence meets.
\begin{lemma}\label{P:ctrioidprops}
Let $S$ be a c-trioid. Then 
  \begin{enumerate}
\item $z\le x\wedge z\le y\Leftrightarrow z\le 
  x\cdot y$, for all $x,y,z\in d(S)$,
\item $z\le x\wedge z\le y \Leftrightarrow z\le x \para y$, for all $x,y,z\in\T(S)$. 
\end{enumerate}
\end{lemma}


\section{Concurrent Dynamic Algebra}\label{S:cda}

This section explains the relationship between c-trioids and
concurrent dynamic algebras, as formalised in~\cite{FurusawaS14}. In
every proto-monoid or proto-trioid $S$ with a domain operation, a
modal diamond operation can be defined as
\begin{equation*}
  \langle x\rangle p = d(x\cdot p)
\end{equation*}
for all $x\in S$ and $p\in d(S)$.  In the setting of c-trioids, some,
but not all, of the concurrent dynamic algebra axioms can be derived.

Here we call \emph{strong c-trioid} a c-trioid $S$ which satisfies,
for all $x,y\in S$ and $p\in\Su(S)$,
\begin{gather*}
  (x\para y) \cdot p = (x\cdot p)\para (y\cdot p),\\
(p\cdot
  x)\cdot y = p\cdot (x\cdot y), \qquad (x\cdot p)\cdot y= x\cdot (p\cdot y),\qquad(x\cdot 
y)\cdot p = x\cdot (y\cdot p).
\end{gather*}
A \emph{strong c-Kleene algebra} is
a strong c-trioid expanded by a star operation which satisfies, for all $x,y\in S$ and $p\in\Su(S)$,
\begin{equation*}
  \sid +x\cdot x^\ast \le x^\ast,\qquad p+x\cdot y\le 
  y\Rightarrow x^\ast \cdot p\le y,\qquad (x\cdot p)\cdot z = x\cdot 
  (p\cdot z).
\end{equation*}
Soundness has been shown in~\cite{FurusawaS14}.
\begin{lemma}
  In every strong c-trioid, the following concurrent dynamic
  algebra axioms are derivable.
  \begin{enumerate}
  \item  $\langle x+y\rangle p = \langle x\rangle p+\langle y\rangle p$,
\item   $\langle x\cdot y\rangle p = \langle x\rangle\langle y\rangle p$,
\item $\langle p\rangle q = p\cdot q$,
\item $\langle x\para  y\rangle p = \langle x\rangle p\cdot \langle y\rangle p$.
  \end{enumerate}
\end{lemma}
\begin{lemma}
  In every strong c-Kleene algebra, the remaining concurrent dynamic
  algebra axioms are derivable as well.
  \begin{enumerate}
  \item  $p+\langle x\rangle\langle x^\ast\rangle p= \langle x^\ast \rangle 
  p$,
\item $ \langle x\rangle p\le p\Rightarrow \langle x^\ast\rangle p \le p$.
  \end{enumerate}
\end{lemma}
Hence every strong c-Kleene algebra is a concurrent dynamic
algebra. The resulting subclass relationships are shown in
Figure~\ref{F:cdaclasses}.
\begin{figure}[h]
  \centering 
\begin{equation*}
  \def\labelstyle{\normalsize}
\xymatrix @R=1pc @C=.2pc{
&&\mathsf{\phantom{c}cKA}^+\\
&\mathsf{dpbKA}\ar@{-}[ur]&& \mathsf{cT}^+\ar@{-}[ul]\\
\mathsf{\phantom{d}dpKA}\ar@{-}[ur]&& \mathsf{dpT}\ar@{-}[ur]\ar@{-}[ul]&& \mathsf{\phantom{c}cT\phantom{i}}\ar@{-}[ul]\\
&\mathsf{dpD}\ar@{-}[ur]\ar@{-}[ul] &&\mathsf{\phantom{pi}pT\phantom{Ai}}\ar@{-}[ul]\ar@{-}[ur]\\
&& \mathsf{pD}\ar@{-}[ur]\ar@{-}[ul]
}
\end{equation*}
  \caption{Class inclusions including domain proto-dioids ($\mathsf{dpD}$),
  domain proto-trioids ($\mathsf{DpT}$),  strong c-trioids 
  ($\mathsf{cT}^+$), domain proto-Kleene algebras ($\mathsf{dpKA}$),
  domain proto-bi-Kleene algebras ($\mathsf{dpbKA}$) and strong c-Kleene 
  algebras ($\mathsf{cKA}^+$)}
\label{F:cdaclasses}
\end{figure}
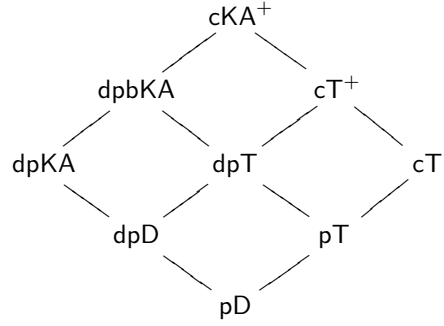

The following counterexample shows that the additional axioms are
necessary.

\begin{lemma}\label{P:ctrioidcounter1}
  In some c-trioid,
  \begin{enumerate}
  \item $(x\para  y)\cdot d(z) \neq (x\cdot d(z))\para (y\cdot d(z))$,
  \item $(x\cdot y)\cdot d(z) \neq x\cdot (y\cdot d(z))$,
  \item   $\langle x\cdot y\rangle p \neq \langle x\rangle\langle 
    y\rangle p$,
  \item $\langle x\para y\rangle p \neq \langle x\rangle p \cdot 
    \langle y\rangle p$. 
  \end{enumerate}
\end{lemma}
\begin{proof}
\begin{enumerate}
\item Let $S=\{a\}$ and let the trioid operations be defined by $0<
  \pid <\sid<a$ and the tables for $\para $ and $\cdot$; from which 
  $d$ can be computed. 
    \begin{equation*}
      \begin{array}{c|cccc}
\para   &0&\pid&\sid & a\\
\hline 
 0  & 0 & 0 & 0 &0\\
\pid &  0& \pid & \sid &  a\\
\sid & 0 & \sid & \sid & a\\
a & 0 & a & a & a 
      \end{array}
\qquad\qquad 
     \begin{array}{c|cccc}
\cdot  &0&\pid&\sid & a\\
\hline 
 0  & 0 & 0 & 0 &0\\
\pid &  0& \pid & \pid &  \pid\\
\sid & 0 & \pid & \sid & a\\
a & \pid & \pid & a & a 
      \end{array}
\qquad\qquad 
\begin{array}[ ]{c|c}
  &d\\
\hline 
0 & 0\\
\pid & \sid\\
\sid & \sid\\
a & \sid 
\end{array}
    \end{equation*}
Then $(a\para  \pid)\cdot d(0) = a\cdot 0= \pid \neq 0= \pid\para  0=
  (a\cdot 0)\para  (\pid\cdot 0)=(a\cdot d(0))\para  (\pid\cdot d(0))$. 
\item With the same counterexample as in (1), $(a\cdot \pid)\cdot 0 =
  \pid\cdot 0 = 0\neq \pid= a \cdot 0 = a\cdot (\pid\cdot 0)$. 
\item Again, with the same counterexample, $\langle a\cdot \pid\rangle
  0 = d((a\cdot \pid)\cdot 0)=d(0)=0\neq \sid = d(\pid) = d(a \cdot
  d(\pid\cdot 0)) = \langle a\rangle \langle \pid\rangle0$.
\item Once more with the same counterexample, $\langle a\para
  \pid\rangle 0=d((a\para \pid)\cdot 0)= d(\pid)=\sid \neq
  0=d(0)=d((a\cdot 0)\para (\pid\cdot 0)) = \langle a\rangle 0\cdot
  \langle \pid\rangle 0$.
\end{enumerate}
\end{proof}


\section{c-Lattices}\label{S:clattices}

We have seen that c-monoids and c-trioids do not capture all
isomorphisms between sequential subidentities, parallel subidentities
and vectors outlined in Section~\ref{S:seqsubids}. In addition, they
are too weak to link the c-monoid-based domain definition with the
alternative ones presented in the same section. A meet
operation, and more specifically a bounded distributive lattice
structure, is needed for this purpose.

A \emph{c-lattice} is a structure 
$(S,+,\sqcap,\cdot,\para ,0,\sid,\pid,U,\pidc)$ such that 
$(L,+,\sqcap,0,U)$ is a bounded distributive lattice with least 
element $0$ and greatest element $U$,
$(S,+,\cdot,\para ,0,\sid,\pid)$ is a proto-trioid
and the following axioms hold.
\begin{align}
x\cdot \pid + x\cdot\pidc &= x\cdot U,\label{eq:cl1}\tag{cl1}\\
\pid\sqcap (x+\pidc)&= x\cdot 0,\label{eq:cl2}\tag{cl2}\\
x \cdot (y \para z) &\le (x \cdot y) \para (x \cdot z),\label{eq:cl3}\tag{cl3}\\
z \para z \le z &\Rightarrow (x \para y) \cdot z = (x \cdot z) \para
(y \cdot z),\label{eq:cl4}\tag{cl4}\\
x \cdot (y \cdot (z \cdot 0)) &= (x \cdot y) \cdot (z \cdot
0),\label{eq:cl5}\tag{cl5}\\
(x\cdot 0)\cdot y &= x\cdot (0\cdot y),\label{eq:cl6}\tag{cl6}\\
\sid \para \sid & =\sid,\label{eq:cl7}\tag{cl7}\\
((x \cdot \pid) \para \sid)\cdot y &= (x\cdot \pid) \para y,\label{eq:cl8}\tag{cl8}\\
((x\sqcap \sid)\cdot \pid)\para \sid &= x\sqcap \sid,\label{eq:cl9}\tag{cl9}\\
((x\sqcap \pidc)\cdot \pid)\para\sid &= \sid\sqcap (x\sqcap\pidc)\cdot \pidc,\label{eq:cl10}\tag{cl10}\\
((x\sqcap\pidc)\cdot\pid)\para\pidc &= (x\sqcap\pidc)\cdot \pidc.\label{eq:cl11}\tag{cl11}
\end{align}
Axiom (\ref{eq:cl1}) and (\ref{eq:cl2}) imply that $\pid$ and $\pidc$
are complements, that is,
\begin{equation*}
  \pid + \pidc = U,\qquad \pid\sqcap \pidc = 0.
\end{equation*}
In addition, (\ref{eq:cl2}) implies that $x\sqcap\pid = x\cdot 0$.

The axioms (\ref{eq:cl3})-(\ref{eq:cl6}) express associativity and
interaction properties. Further associativity properties for
sequential and parallel subidentities as well as for $U$ hold in the
multirelational model, but those are either derivable or not directly
needed for the main results in this article. We mention them
explicitly whenever they occur. Axioms (\ref{eq:cl7}) and
(\ref{eq:cl8}) are taken from c-monoids. Axiom (\ref{eq:cl9}) can be
written as $d(x\sqcap \sid)=x\sqcap\sid$; it expresses one of the
isomorphism conditions from Figure~\ref{F:extendeddiagram}. Axiom
(\ref{eq:cl10}) can be written as $d(x\sqcap\pidc)= \sid\sqcap
(x\sqcap\pidc)\cdot\pidc$.  Axiom (\ref{eq:cl11}) can be rewritten as
$d(x\sqcap \pidc)\cdot \pidc = (x\sqcap\pidc)\cdot \pidc$. These
domain properties are reminiscent of the relational case and have been
motivated in Section~\ref{S:seqsubids}.

We have used Isabelle's counterexample generators to analyse
irredundancy of these axioms. However, due to their large number, this
was not always successful. Whether the set of axioms can be compacted
further remains to be seen.

\begin{proposition}\label{P:clatticectrioid}
  Every c-lattice is a c-trioid. 
\end{proposition}
It follows that every c-lattice is a c-monoid. Next we prove a
soundness result.

 \begin{proposition}\label{P:mrclattice}
    $(\M(X),\cap,\cdot,\para ,\emptyset,\sid,\pid,U,\pidc)$ forms a c-lattice.
 \end{proposition}
 \begin{proof}
   We have verified the proto-trioid axioms in
   Proposition~\ref{P:domtrioid}, Axioms (\ref{eq:cl3})-(\ref{eq:cl6})
   in Lemma~\ref{P:interactionlemma}, Axioms (\ref{eq:cl7}) and
   (\ref{eq:cl8}) in Proposition~\ref{P:dommonoid}, and the remaining
   axioms in Lemma~\ref{P:mraxioms}.
  \end{proof}

  We call a c-lattice \emph{boolean} if its lattice reduct forms a
  boolean algebra.  It is easy to see that multirelations form in fact
  boolean c-lattices.  Moreover, infima and suprema of arbitrary sets
  exist in the algebra of multirelations, and an infinite left
  distributivity law for sequential composition with respect to
  suprema and infinite left and right distributivity laws for
  parallel composition with respect to suprema hold. Multirelations
  therefore form quantale-like algebras. This is further explored in
  Sections~\ref{S:finiteiteration} and~\ref{S:divergence}.

For a c-lattice $S$ we can now define also
\begin{equation*}
  \V(S)=\{x\in S\mid (x\cdot\pid)\para U = x\},\qquad \U(S)=\{x\in
  S\mid x \para U = x\}.
\end{equation*}
\begin{lemma}\label{P:subidfix}
  Let $S$ be a c-lattice. Then
  \begin{enumerate}
  \item $d(S)=\Su(S)$,
  \item $\T(S)=\{x\in S\mid x\le \pid\}=\{x\in S\mid x\cdot
    0=x\}$,
\item $\V(S)=\{x\in S\mid d(x)\cdot U = x\}$,
 \item $\N(S)=\{x\in S\mid x\sqcap \pid =0\}=\{x\in S\mid x\le \pidc\}$.
  \end{enumerate}
\end{lemma}
More specifically, $d(x)=x$ if and only if $x\le \sid$, $x\cdot\pid =
x$ is equivalent to each of $x\le \pid$ and $x\cdot 0 = x$, $(x\cdot
c)\para U = x$ if and only if $d(x)\cdot U =x$, and $x\cdot 0 = 0$ is
equivalent to each of $x\sqcap \pid = 0$ and $x\le \pidc$.  This
correspondence between subidentities, domain elements and terminal
elements captures that in the multirelational model.

The next three lemmas collect some basic properties of c-lattices.  
\begin{lemma}\label{P:clatprops1}
Let $S$ be a c-lattice. Then
  \begin{enumerate}
\item $(x\cdot y)\cdot z = x\cdot (y\cdot z)$, if $z\in\T(S)$,
\item $(x\cdot z)\para (y\cdot z) = (x\para y)\cdot z$, if $z\in \T(S)\cup\Su(S)$,
\item $x\le x\para x$,
\item $x\sqcap y\le x\para y$.
\end{enumerate}
\end{lemma}

\begin{lemma}\label{P:clatprops1a}
Let $S$ be a  c-lattice. Then
\begin{enumerate}
\item $x= (x\sqcap\pidc)+x\cdot 0$,
\item $x\le \pidc$, if $x\in\Su(S)$,
\item $(x\cdot\pid)\sqcap\pidc = (x\cdot 0)\sqcap \pidc = (x\sqcap\pidc)\cdot 0 = 0$.
 \end{enumerate}
\end{lemma}

\begin{lemma}\label{P:clatprops2}
  In every c-lattice,
  \begin{enumerate}
\item $\pid = U\cdot 0$,
\item $U\para U = U\cdot U= U\cdot \pidc = \pidc \cdot U = U$,
\item $\pid\cdot x = \pidc \cdot \pid = U \cdot\pid = \pid$,
\item $\pidc\para\pidc = U\para \pidc = \pidc\cdot\pidc= \pidc$.
  \end{enumerate}
\end{lemma}
Property (1) gives an explicit definition of $\pid$. Some of the other
properties have already been stated in Lemma~\ref{P:constprops} in the
multirelational model. It follows that all constants are
sequential and parallel idempotents.

The final lemma of this section collects further simplification and
decomposition properties that are interesting for domain definitions.
\begin{lemma}\label{P:clatticeprops3}
  In every c-lattice,
  \begin{enumerate}
  \item $x\cdot y = (x\sqcap\pidc)\cdot y + x\cdot 0$,
\item $\sid\sqcap (x\sqcap\pidc)\cdot y = \sid\sqcap x\cdot y$,
\item $\sid\sqcap x\cdot\pidc = \sid\sqcap x\cdot U$,
\item $\sid\sqcap x\para\pidc = \sid\sqcap x\para U$,
\item $(x\cdot \pid)\para\pidc = (x\sqcap\pidc)\cdot\pidc + (x\cdot 0)\para\pidc$,
\item $(x\cdot \pid)\para U= x\cdot U+(x\cdot 0)\para U$.
\end{enumerate}
\end{lemma}

Equations (2), (3) and (4) can be visualised by the following
subdiagrams of Figure~\ref{F:extendeddiagram}.

\begin{equation*}
  \def\labelstyle{\normalsize}
\xymatrix @R=2pc @C=2pc{
S\ar[r]^{(\sqcap\pidc)} \ar[d]_{(\cdot y)}& \N(S)\ar[r]^{(\cdot y)} & S\ar[d]^{(\sqcap\sid)}\\
S\ar[rr]_{(\sqcap \sid)}& &\Su(S) 
}
\quad 
  \def\labelstyle{\normalsize}
\xymatrix @R=2pc @C=2pc{
S\ar[r]^{(\cdot\pidc)}\ar[d]_{(\cdot U)}& S\ar[d]^{(\sqcap\sid)}\\
S\ar[r]_{(\sqcap\sid)} & \Su(S) 
}
\quad 
  \def\labelstyle{\normalsize}
\xymatrix @R=2pc @C=2pc{
S\ar[r]^{(\para\pidc)}\ar[d]_{(\para U)}& S\ar[d]^{(\sqcap\sid)}\\
S\ar[r]_{(\sqcap\sid)} & \Su(S) 
}
\end{equation*}


\section{Domain in c-Lattices}\label{S:clatticedomain}

This section presents explicit domain definitions and related 
properties for general and nonterminal elements. First, using 
$d(x)=(x\cdot\pid)\para\sid$, recall that we can rewrite some of the c-lattice 
axioms:
\begin{align*}
  d(x)\cdot y &= (x\cdot\pid)\para y,\tag{cl8}\\
d(x\sqcap\sid) &=x\sqcap\sid,\tag{cl9}\\
d(x\sqcap\pidc)&=\sid\sqcap(x\sqcap\pidc)\cdot\pidc,\tag{cl10}\\
d(x\sqcap\pidc)\cdot\pidc &= (x\sqcap\pidc)\cdot\pidc.\tag{cl11}
\end{align*}
Properties (\ref{eq:cl10}) and (\ref{eq:cl11}) are visualised by the
following subdiagram of Figure~\ref{F:extendeddiagram}.
\begin{equation*}
  \def\labelstyle{\normalsize}
\xymatrix @R=2pc @C=4pc{
\N(S)\ar[r]^{(\cdot\pidc)}\ar[dr]_{d} &\V(\N(S))\ar@/^/[d]^{(\sqcap\sid)}\\
& \Su(S)\ar@/^/[u]^{(\cdot\pidc)}
}
\end{equation*}
These two identities admit the following variations.
\begin{lemma}\label{P:dclatprop1}
In every c-lattice,
  \begin{enumerate}
\item $d(U)=d(\pidc)=\sid$,
  \item $d(x\sqcap\pidc) =\sid\sqcap x\cdot\pidc$,
\item $d(x\sqcap\pidc) =\sid\sqcap (x \sqcap \pidc)\cdot U$,
\item $d(x\sqcap\pidc) =\sid\sqcap x\cdot U$,
\item $d(x\sqcap\pidc) = \sid\sqcap ((x\sqcap \pidc)\cdot 
  \pid)\para\pidc$,
\item $d(x\sqcap\pidc) = \sid\sqcap ((x\sqcap \pidc)\cdot \pid)\para U$. 
  \end{enumerate}
\end{lemma}
Identity (5) corresponds to another subdiagram of
Figure~\ref{F:extendeddiagram}.
\begin{equation*}
  \def\labelstyle{\normalsize}
\xymatrix @R=2pc @C=2pc{
\N(S)\ar[r]^{(\cdot\pid)}\ar[d]_d& \T(\N(S))\ar[d]^{(\para\pidc)}\\
\Su(S) & \V(\N(S))\ar[l]^{(\sqcap\sid)}
}
\end{equation*}
The remaining properties are obtained by combining the diagrams for
(\ref{eq:cl10}), (\ref{eq:cl11}) and (5) with those for
Lemma~\ref{P:clatticeprops3}. Equation (3), for instance, corresponds
to the following commuting diagram.
\begin{equation*}
  \def\labelstyle{\normalsize}
\xymatrix @R=2pc @C=2pc{
\N(S)\ar[r]^{(\cdot\pidc)}\ar[dr]_d\ar[d]_{(\cdot U)}& \V(\N(S))\ar[d]^{(\sqcap\sid)}\\
S\ar[r]_{(\sqcap\sid)} & \Su(S)
}
\end{equation*}

Since $d(x)=d(x\sqcap\pidc)+d(x\cdot 0)$ by Lemma~\ref{P:clatprops1}
and Proposition~\ref{P:cprototrioid}, we can  generate explicit 
definitions of $d(x)$ by inserting those for 
$d(x\sqcap\pidc)$ and $d(x\cdot 0)$, for instance 
\begin{equation*}
  d(x)= (\sid\sqcap x\cdot U)+(x\cdot 0)\para\sid. 
\end{equation*}
This explains, in particular, the difference between relational and 
multirelational domain definitions. We can also use the various 
definitions of $d(x\sqcap\pidc)$ and $d(x\cdot 0)$ to generate more 
compact definitions. 

\begin{lemma}\label{P:explicitdomverif}
In every c-lattice,
  \begin{enumerate}
  \item $d(x) = \sid\sqcap (x\cdot \pid)\para \pidc$,
 \item $d(x) = \sid\sqcap (x\cdot \pid)\para U$,
  \item $d(x)= (\pid\sqcap x\cdot U)\para \sid$. 
\end{enumerate}
\end{lemma}
These identities correspond to the following commuting subdiagrams of
Figure~\ref{F:extendeddiagram}.
\begin{equation*}
  \def\labelstyle{\normalsize}
\xymatrix @R=2pc @C=5pc{
 S\ar[d]_{(\para\pidc)\circ(\cdot\pid)}\ar[dr]_{d}\ar[r]^{(\para U)\circ (\cdot \pid)} &
 \V(S)\ar[d]^{(\sqcap \sid)}\\
S\ar[r]_{(\sqcap\sid)}& \Su(S) 
}
\qquad\quad 
  \def\labelstyle{\normalsize}
\xymatrix @R=2pc @C=3pc{
 S\ar[dr]^{(\cdot \pid)}\ar[r]^{(\cdot U)} \ar[d]_{d}&
 S\ar[d]^{(\sqcap \pid)}\\
\Su(S)& \T(S) \ar[l]^{(\para  \sid)}
}
\end{equation*}
Finally we mention some variations of Axiom (\ref{eq:cl11}). 
\begin{lemma}\label{P:dompropsverif} 
  In every c-lattice,
  \begin{enumerate}
\item $d(x\sqcap\pidc)\cdot U = (x\sqcap\pidc)\cdot U$,
\item $d(x)\cdot\pidc = (x\sqcap\pidc)\cdot\pidc + (x\cdot 
  0)\para\pidc$,
\item $d(x)\cdot U = (x\sqcap\pidc)\cdot U + (x\cdot 0)\para U$,
  \item $d(x)\cdot U = x\cdot U+(x\cdot 0)\para U$,
\item $x\cdot \pidc =  d(x\sqcap\pidc)\cdot \pidc+ x\cdot 0$. 
\item $x\cdot U= d(x\sqcap\pidc)\cdot U+ x\cdot 0$,
\item $d(x\cdot U) = d(x\cdot\pidc) = d(x)$. 
\end{enumerate}
\end{lemma}


\section{Subalgebras of c-Lattices}\label{S:clatticesubalgs}

This section studies the structure of the subalgebras of sequential
subidentities, parallel subidentities, vectors and nonterminal
elements.  First we consider the set $\Su(S)$ of sequential
subidentities of a c-lattice $S$. We can freely identify subidentities
with domain elements and use results for $d(S)$ from
Section~\ref{S:clatticedomain}.
\begin{proposition}\label{P:dsubalgebra}
  Let $S$ be a c-lattice.  The set $\Su(S)$ forms a distributive lattice
  bounded by $0$ and $\sid$, and in which sequential and parallel
  composition are meet. It forms a boolean algebra if $S$ is
  boolean.
\end{proposition}
\begin{proof}
  Relative to Proposition~\ref{P:ctrioidbdl} it remains to check that
  $\Su(S)$ is closed under meets and that meet coincides with
  multiplication in this subalgebra.  Meet closure follows from
  $d(p\sqcap q)=p\sqcap q$ for all $p,q\in\Su(S)$. Sequential
  composition of domain elements is a greatest lower bound operation
  in $\Su(S)$ by Lemma~\ref{P:ctrioidprops}(4). Hence it coincides
  with meet and parallel composition. Finally, if $S$ is a boolean
  algebra with complement $\overline{x}$ for each $x\in S$, we define
  complementation on $\Su(S)$ by $x'=\sid\sqcap\overline{x}$.
\end{proof}

\begin{proposition}\label{P:psubidalg}
  Let $S$ be a c-lattice. The set $\T(S)$ forms a sub-c-lattice
  bounded by $0$ and $\pid$, in which parallel composition is meet,
  and in which all $x\in \T(S)$ are right units of sequential
  composition. It is boolean whenever $S$ is.
\end{proposition}
\begin{proof}
  First, parallel composition of terminal elements is a greatest
  lower bound operation on $\T(S)$ by
  Lemma~\ref{P:ctrioidprops}(7). It therefore coincides with
  meet. Second, the c-lattice operations are closed: $0\cdot 0 = 0$,
  $\pid\cdot 0 = \pid$, $x\cdot 0 + y\cdot 0 = (x+y)\cdot 0$,
 $(x\cdot 0)\sqcap (y\cdot 0) = (x\sqcap y)\cdot 0$,
 $(x\cdot 0)\para (y\cdot 0) = (x\para y)\cdot 0$ and $(x\cdot 0)\cdot (y\cdot 0) = x\cdot 0$. 

 Thus $\T(S)$ forms a c-lattice bounded by $0$ and $\pid$. By $(x\cdot
 0)\cdot (y\cdot 0)=x\cdot 0$, sequential composition is a projection,
 hence every terminal element  a right identity. Finally, for every
 boolean c-lattice $S$, the set $\T(S)$ forms a boolean subalgebra
 with complementation defined similarly to the sequential case.
\end{proof}

The subset $\V(S)$ of vectors of a c-lattice $S$ does not have such
pleasant properties.  In particular, vectors are not closed under
sequential composition.  With $x\in\V(S)$ if and only if $d(x)\cdot U=
x$ it is easy to see that $U\in\V(S)$ and $0\in\V(S)$, whereas $U\cdot
0=\pid\not\in\V(S)$. In addition, vectors in c-lattices need not be
closed under meets, whereas this is the case in the multirelational
model, where $(R\cap S)\cdot T= R\cdot T\cap S\cdot T$ holds for
$R,S\in\Su(X)$. At least we have the following closure properties.
\begin{lemma}\label{P:Vaux}
  In every c-lattice,
  \begin{enumerate}
  \item $d(0)\cdot U = 0$ and $d(U)\cdot U = U$,
  \item $d(x)\cdot z+d(y)\cdot z = d(x+y)\cdot z$,
  \item $(d(x)\cdot U)\para (d(y)\cdot U) = d(x\para y)\cdot U$.
   \end{enumerate}
\end{lemma}

Finally we consider the subset $\N(S)$ of nonterminal elements, where
the analogy with binary relations is more striking.

\begin{proposition}\label{P:Nsubalg}
  Let $S$ be a c-lattice.  The set $\N(S)$ forms a sub-c-lattice of
  $S$ without parallel unit in which $0$ is a right annihilator of
  sequential composition and $\pidc$ the maximal element. It is
  boolean whenever $S$ is.
\end{proposition}
\begin{proof}
  It needs to be checked that $0\sqcap \pidc= 0$, $\sid\sqcap \pidc =
  \sid$ and that $(x+y)\sqcap \pidc = x +y$, $(x\sqcap y)\sqcap \pidc
  = x \sqcap y$, $(x\cdot y)\sqcap \pidc = x \cdot y$, $(x\para
  y)\sqcap \pidc = x \para y$ hold for all $x,y\in\N(S)$.

Also, in every non-trivial algebra, $\pid\not\in \N(S)$.  Finally,
$(x\sqcap\pidc)\cdot 0 = 0$ by Lemma~\ref{P:clatprops1a}(3) and $\pidc$
is the maximal element by definition.
\end{proof}

Next we consider the subalgebras $\Su(\N(S))$ and $\V(\N(S))$. First,
$\Su(S)\subseteq \N(S)$, hence $\Su(\N(S))=\Su(S)$ and
Proposition~\ref{P:dsubalgebra} applies without modification to
$\N(S)$.
\begin{corollary}\label{P:Nsubalgcor}
  Let $S$ be a c-lattice.  The set $\Su(\N(S))=\Su(S)$ forms a bounded
  distributive sublattice of $\N(S)$. It is a boolean algebra whenever
  $S$ is.
\end{corollary}

Next we consider the subalgebra of vectors. First we derive a variant
of \emph{Tarski's rule} from relation algebra (cf.~\cite{Maddux}).
\begin{lemma}\label{P:tarski}
Let $R\in \M(X)$. Then 
$    R \cap \pidc\neq\emptyset \ \Rightarrow \ \pidc \cdot  ((R
    \cap\pidc) \cdot \pidc) = \pidc$.
\end{lemma}
As common in relation algebra we keep Tarski's rule separate from the other
axioms since it is not even a quasi-identity. We can use it to prove
the following fact.
\begin{lemma}\label{P:vecprop}
  Let $S$ be a c-lattice in which Tarski's rule and $d(x)\cdot (y\cdot
  z)=(d(x)\cdot y)\cdot z$ hold. Let $x,y\in \V(\N(S))$. Then
  \begin{equation*}
    x\cdot y = 
    \begin{cases}
      0, &\text{if } y = 0,\\
      x, &\text{if } y\neq 0.
    \end{cases}
  \end{equation*}
\end{lemma}
\begin{proposition}\label{P:VNsubalg}
  Let $S$ be a c-lattice in which Tarski's rule and 
  \begin{equation*}
    (d(x)\sqcap d(y))\cdot z = d(x)\cdot z \sqcap d(y)\cdot z,\qquad
    (d(x)\cdot y)\cdot z =d(x)\cdot (y\cdot z)
  \end{equation*}
  hold. Then $\V(\N(S))$ is a sub-c-lattice of $\N(S)$ bounded by
  $0$ and  $\pidc$, in which $0$ is a left annihilator
  and parallel composition is meet.
\end{proposition}
\begin{proof}
  We need to verify the closure conditions $0\cdot \pidc = 0$ and
  $\pidc\cdot\pidc = \pidc$ as well as $(x+y)\cdot\pidc = x + y$,
  $(x\sqcap y)\cdot \pidc = x\sqcap y$, $(x\cdot y)\cdot\pidc = x\cdot
  y$ and $(x\para y)\cdot \pidc = x\para y$ for all $x,y\in \N(S)$,
  and $x\para y=x\sqcap y$ for all $x,y\in \N(S)$.
\end{proof}
Note that meet-closure is enforced by assuming $(d(x)\sqcap d(y))\cdot
z = d(x)\cdot z \sqcap d(y)\cdot z$. This and all additional
assumptions on multirelations used in this section have, of course,
been verified in the multirelational model. Whether stronger
properties hold in situations where sequential composition is
associative remains to be seen.


\section{Isomorphisms in c-Lattices}\label{S:clatticeiso}

This section finally verifies the isomorphisms between sequential
subidentities, parallel subidentities and vectors from
Figure~\ref{F:magictriangle} and~\ref{F:extendeddiagram} in the
context of c-lattices and for their nonterminal elements.  

We also characterise the structure that is preserved by these
mappings. Given the results on subalgebras from the previous section
it cannot be expected that sequential composition is preserved. This
is indeed confirmed by the multirelational counterexamples in this
section---apart from one exception.  The other c-lattice operations
are preserved. In particular, all isomorphisms
are constructed from the operations and constants of c-lattices, so
that their properties can be checked within the c-lattice setting by
simple equational reasoning.

\begin{proposition}\label{P:iso2}
  Let $S$ be a c-lattice. 
\begin{enumerate}
\item The maps $(\cdot U)$ and $d$ as well as
  $(\para U)$ and $(\cdot\pid)$ in the following diagrams form
  bijective pairs; the  diagrams commute.
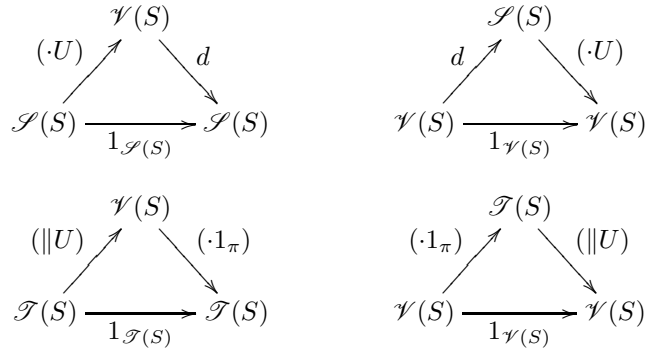
\begin{equation*}
  \def\labelstyle{\normalsize}
\xymatrix @R=2pc @C=.5pc{
&\V(S)\ar[dr]^{d}\\
\Su(S)\ar[ur]^{(\cdot U)}\ar[rr]_{1_{\Su(S)}}&&\Su(S) 
}
\qquad\qquad 
\xymatrix @R=2pc @C=.5pc{
&\Su(S)\ar[dr]^{(\cdot U)}\\
\V(S)\ar[ur]^{d}\ar[rr]_{1_{\V(S)}}&&\V(S) 
}
\end{equation*}
\begin{equation*}
  \def\labelstyle{\normalsize}
\xymatrix @R=2pc @C=.5pc{
&\V(S)\ar[dr]^{(\cdot\pid)}\\
\T(S)\ar[ur]^{(\para U)}\ar[rr]_{1_{\T(S)}}&&\T(S) 
}
\qquad\qquad 
\xymatrix @R=2pc @C=.5pc{
&\T(S)\ar[dr]^{(\para U)}\\
\V(S)\ar[ur]^{(\cdot\pid)}\ar[rr]_{1_{\V(S)}}&&\V(S) 
}
\end{equation*}
\item Therefore $\Su(S)\cong \T(S)\cong\V(S)$. 
\end{enumerate}
\end{proposition}
\begin{proof}
  The isomorphism between $d(S)$ and $\T(S)$ has been verified in
  Proposition~\ref{P:dtermiso}; moreover $d(S)=\Su(S)$ by
  Lemma~\ref{P:subidfix}.  It remains to check that
  \begin{align*}
    d(d(x)\cdot U) &= d(x),&
    d((x\cdot\pid)\para U)\cdot U &= (x\cdot\pid)\para U,\\
((x\cdot\pid)\para U)\cdot\pid &= x\cdot\pid,&
(((x\cdot\pid)\para U)\cdot\pid)\para U &= (x\cdot\pid)\para U. 
  \end{align*}
\end{proof}
\begin{proposition}\label{P:iso3}
  Let $S$ be a c-lattice.  
\begin{enumerate}
\item The maps $(\cdot\pid)$ and $d$, 
  $(\cdot\pidc)$ and $d$, and $(\para\pidc)$ and $(\cdot\pid)$ in the
  following diagrams form bijective pairs; the diagrams
  commute.
\begin{equation*}
  \def\labelstyle{\normalsize}
\xymatrix @R=2pc @C=.5pc{
&\T(\N(S))\ar[dr]^{d}\\
\Su(\N(S))\ar[ur]^{(\cdot \pid)}\ar[rr]_{1_{\Su(\N(S))}}&&\Su(\N(S)) 
}
\qquad\qquad 
\xymatrix @R=2pc @C=.5pc{
&\Su(\N(S))\ar[dr]^{(\cdot \pid)}\\
\T(\N(S))\ar[ur]^{d}\ar[rr]_{1_{\T(\N(S))}}&&\T(\N(S)) 
}
\end{equation*}
 \begin{equation*}
  \def\labelstyle{\normalsize}
\xymatrix @R=2pc @C=.5pc{
&\V(\N(S))\ar[dr]^{d}\\
\Su(\N(S))\ar[ur]^{(\cdot \pidc)}\ar[rr]_{1_{\Su(\N(S))}}&&\Su(\N(S)) 
}
\qquad\qquad 
\xymatrix @R=2pc @C=.5pc{
&\Su(\N(S))\ar[dr]^{(\cdot \pidc)}\\
\V(\N(S))\ar[ur]^{d}\ar[rr]_{1_{\V(\N(S))}}&&\V(\N(S)) 
}
\end{equation*}
\begin{equation*}
  \def\labelstyle{\normalsize}
\xymatrix @R=2pc @C=.5pc{
&\V(\N(S))\ar[dr]^{(\cdot\pid)}\\
\T(\N(S))\ar[ur]^{(\para \pidc)}\ar[rr]_{1_{\T(\N(S))}}&&\T(S) 
}
\qquad\qquad 
\xymatrix @R=2pc @C=.5pc{
&\T(\N(S))\ar[dr]^{(\para \pidc)}\\
\V(\N(S))\ar[ur]^{(\cdot\pid)}\ar[rr]_{1_{\V(\N(S))}}&&\V(\N(S)) 
}
\end{equation*}
\item Therefore $\Su(\N(S))\cong \T(\N(S))\cong\V(\N(S))$. 
\end{enumerate}
\end{proposition}
\begin{proof}
The following conditions must be checked.
\begin{align*}
  d(d(x \sqcap \pidc)\cdot\pid &= d(x\sqcap\pidc),
& d((x\sqcap\pidc)\cdot\pid)\cdot \pid &= (x\sqcap \pidc)\cdot\pid,\\
    d(d(x\sqcap\pidc)\cdot \pidc) &= d(x\sqcap\pidc),&
d((x \sqcap \pidc)\cdot\pidc)\cdot\pidc &= (x\sqcap\pidc)\cdot\pidc,\\
 (((x\sqcap\pidc)\cdot\pid)\para \pidc)\cdot\pid &= (x\sqcap\pidc)\cdot\pid,&
(((x\sqcap\pidc)\cdot\pidc)\cdot\pid)\para \pidc &= (x\sqcap\pidc) \cdot\pidc.
\end{align*}
\end{proof}

We now investigate structure preservation of these bijections.

\begin{proposition}\label{P:structurepreservation}
  Let $S$ be a c-lattice.  The maps
  \begin{align*}
 (\cdot\pid) &: \Su(S)\to \T(S),&d&:\T(S)\to \Su(S),\\
(\cdot U)&:\Su(S)\to \V(S),&d&:\V(S)\to \Su(S),\\
(\para U)&:\T(S)\to \V(S), &(\cdot\pid)&:\V(S)\to \T(S) 
  \end{align*}
  preserve addition, meet and parallel composition, minimal elements
  and maximal elements of the subalgebras. The last map also preserves
  sequential composition.
 \end{proposition}

\begin{proposition}\label{P:nstructurepreservation}
  Let $S$ be a c-lattice.  The maps
  \begin{align*}
 (\cdot\pid) &: \Su(\N(S))\to \T(\N(S))&
d&:\T(\N(S))\to \Su(\N(S))\\
(\cdot \pidc)&:\Su(\N(S))\to \V(\N(S))&
d&:\V(\N(S))\to \Su(\N(S))\\
(\para \pidc)&:\T(\N(S))\to \V(\N(S)) &
(\cdot\pid)&:\V(\N(S))\to \T(\N(S)) 
  \end{align*}
  preserve addition, meet and parallel composition, minimal elements
  and maximal elements of the subalgebras. 
\end{proposition}

We complement these results by refuting preservation of sequential
composition for the remaining maps between sequential identities,
subidentities and vectors.

\begin{lemma}
   \begin{enumerate}
  \item No isomorphism in $\M(X)$ except $(\cdot \pid)$
    preserves sequential composition.
\item No isomorphism in $\N(X)$ preserves sequential composition. 
  \end{enumerate}
\end{lemma}
\begin{proof}
  \begin{itemize}
  \item Let $R=\{(a,\{a\})\}$ and $S=\emptyset$, both of which are in
    $\Su(\N(X))$ and $\V(\N(X))$.  Then
\begin{gather*}
(R\cdot S)\cdot\pid = \emptyset \subset
    \{(a,\emptyset)\} = (R\cdot\pid)\cdot(S\cdot\pid),\\
   (R\cdot S)\cdot U = \emptyset\subset \{(a,\emptyset)\}=(R\cdot 
    U)\cdot \emptyset= (R\cdot U)\cdot (S\cdot U).
 \end{gather*}
\item Let $R=\{(a,\emptyset)\}$ and $S=\emptyset$, both of which are
  in $\T(\N(X))$ and $\V(X)$. Then
  \begin{gather*}
    d(R\cdot S)= d(R)=\{(a,\{a\})\} \supset \emptyset =
    d(R)\cdot\emptyset = d(R)\cdot d(S),\\
   (R\cdot S) \para U= \{(a,\emptyset),(a,\{a\})\} \supset R =
    (R\para U)\cdot\emptyset = (R\para U) \cdot (S\para U).
  \end{gather*}
\item 
Let $R=\{(a,\{a\}),(a,\{b\}),(a,\{a,b\})\}$,
$S=\{(b,\{a\}),(b,\{b\}),(b,\{a,b\})\}$, both of which are in
$\V(\N(X))$. Then
$d(R\cdot S)=d(R)=\{(a,\{a\})\} \supset\emptyset =d(R)\cdot d(S)$.
   \end{itemize}
  \end{proof}


\section{Terminal and Nonterminal Elements}\label{S:nabladelta}

Algebras of multirelations share some features with algebras of
languages with finite and infinite words. Both form trioids with
parallel composition corresponding to shuffle in the language case.
However, $(X\para Y)\cdot Z \subseteq (X\cdot Z)\para (Y\cdot Z)$ does
not generally hold in languages (consider $X=\{a\}$, $Y=\{b\}$ and
$Z=\{cd\}$), the algebra of sequential subidentities is trivial (it
consists only of the empty and the empty word language), and the
notion of vector does not seem to make sense.

For an alphabet $\Sigma$, a finite word language is a subset of
$\Sigma^\ast$, the set of all finite words over $\Sigma$. An infinite
word language is a subset of $\Sigma^\omega$, the set of all strictly
infinite words over $\Sigma$. Languages in which finite and infinite
words are mixed are subsets of
$\Sigma^\infty=\Sigma^\ast\cup\Sigma^\omega$.  One can then use
$\mathsf{fin}:2^{\Sigma^\infty}\to 2^{\Sigma^\ast}$ and
$\mathsf{inf}:2^{\Sigma^\infty}\to 2^{\Sigma^\omega}$ to project on
the finite and infinite words in a language.

Analogously, the maps $\mathsf{fin}$ to $\tau:\M(X)\to 
\T(X)$ and $\mathsf{inf}$ to $\nu:\M(X)\to \N(X)$ defined by
\begin{equation*}
  \tau=\lambda x.\ x\cdot 0,\qquad \nu=\lambda x.\ x\sqcap\pidc
\end{equation*}
project on the terminal and the nonterminal part of a
multirelation. More abstractly, we define such functions $\tau:S\to \T(S)$
and $\nu:S\to \N(S)$ on a c-lattice $S$.

Many properties of $\mathsf{fin}$ and $\mathsf{inf}$ (cf.~\cite{Park}
) are shared with $\nu$ and $\tau$ , but there are also
differences. 

\begin{lemma}\label{P:nutauinterior}
In every c-lattice,
\begin{enumerate}
\item  the functions $\tau$ and $\nu$ are interior operators, and
  therefore retractions,
\item $\tau(x)+\nu(x)=x$ and $\tau(x)\sqcap\nu(x)=0$,
\item $\tau(\nu(x))=0$ and $\nu(\tau(x))=0$.
\end{enumerate}
\end{lemma}
The properties $\tau(x)\le x$, $\tau(\tau(x))=\tau(x)$ and $x\le
y\Rightarrow \tau(x)\le \tau(y)$ must be verified to show that $\tau$
is an interior operator, and likewise for $\nu$. The next lemmas are essentially 
transcriptions of properties verified in the proofs of
Proposition~\ref{P:psubidalg} and \ref{P:Nsubalg}.
\begin{lemma}\label{P:nutauconst}
In every c-lattice,
  \begin{enumerate}
  \item $\tau(0)=0$ and $\nu(0)=0$,
  \item $\tau(\sid)=0$ and $\nu(\sid)=\sid$,
  \item $\tau(\pid)=\pid$ and $\nu(\pid)=0$,
  \item $\tau(\pidc)=0$ and $\nu(\pidc)=\pidc$,
  \item $\tau(U)=\pidc$ and $\nu(U)=\pidc$.
  \end{enumerate}
\end{lemma}
\begin{lemma}\label{P:nutauhoms}
  In every c-lattice,
  \begin{enumerate}
  \item $\tau(x+y)=\tau(x)+\tau(y)$ and $\nu(x+y)=\nu(x)+\nu(y)$,
  \item $\tau(x\sqcap y)=\tau(x)\sqcap \tau(y)$ and $\nu(x\sqcap
    y)=\nu(x)\sqcap\nu(y)$,
  \item $\tau(x\para y) =\tau(x)\para\tau(y)$ and $\nu(x\para y) =
    d(\tau(x))\cdot\nu(y) + d(\tau(y))\cdot\nu(x) +\nu(x)\para \nu(y)$,
\item $\tau(x\cdot y)=\tau(x)+\nu(x)\cdot\tau(y)$.
  \end{enumerate}
\end{lemma}
Lemma~\ref{P:nutauconst} and~\ref{P:nutauhoms} show that $\tau$ and
$\nu$ preserve the constants, addition and
meet.  Moreover, $\tau$ preserves  parallel
composition. The next lemma refutes such a property for the remaining
operations.

\begin{lemma}\label{P:nutauhomcounter}
\begin{enumerate}
  \item There are $R,S\in \M(X)$ such that $\tau(R\cdot S)\neq
  \tau(R)\cdot\tau(S)$.
\item There are $R,S\in \M(X)$ such that $\nu(R\cdot S)\neq
  \nu(R)\cdot \nu(S)$.
\item There are $R,S\in \M(X)$ such that $\nu(R\para S)\neq
  \nu(R)||\nu(S)$.
\end{enumerate}
\end{lemma}
\begin{proof}
\begin{enumerate}
\item If $R=\{(a,\emptyset),(b,\{a\})\}$ and
  $S=\{(a,\emptyset)\}$, then
$\tau(R)\cdot\tau(S)=S\subset  R\cdot\pid=\tau(R\cdot S)$. 
\item  If $R=\{(a,\{a,b\})\}$ and $S=\{(a,\emptyset),(b,\{a,b\})\}$, then 
$  \nu(R\cdot S) = R \neq \emptyset = \nu(R)\cdot\nu(S)$. 
\item If $R=\{(a,\{a\})\}$ and $S=\{(a,\emptyset)\}$, then $\nu(R||S)=R\neq \emptyset = \nu(R)||\nu(S)$. 
\end{enumerate}
\end{proof}
We do not have a compositional characterisation of $\nu(x\cdot y)$.  On the
one hand, $\nu(x\cdot y) = \nu(\nu(x)\cdot y)$, but on the other hand,
without left distributivity, this cannot easily be decomposed further.
In the multirelational model, elements $(b,\emptyset)\in S$ can
obviously contribute to pairs $(a,A)\in R\cdot S$ with
$A\neq\emptyset$. This makes the situation different from the language
case, where $\mathsf{fin}(X\cdot
Y)=\mathsf{fin}(X)\cdot\mathsf{fin}(Y)$. Interestingly, however, this
does not rule out simple decomposition theorems for sequential and
parallel composition.
\begin{lemma}\label{P:taunusplit}
In every c-lattice,
  \begin{enumerate}
  \item $x\cdot y = \tau(x)+\nu(x)\cdot y$,
\item   $x\para y=\nu(x)\para\nu(y)+d(\nu(x))\cdot\tau(y)+d(\nu(y))\cdot\tau(x)+\tau(x)\para\tau(y)$. 
  \end{enumerate}
\end{lemma}

It seems natural to identify elements of c-lattices if they coincide
on their terminal or their nonterminal parts.

\begin{lemma}\label{P:nsideal}
Let $S$ be a c-lattice. Then $\N(S)$ and $\T(S)$ form order ideals.
\end{lemma}
It is straightforward to check that  $x\in \mathscr{F}(S)$ and $y\le x$
  imply $y\in \mathscr{F}(S)$, and that $x,y\in \mathscr{F}(S)$ imply $x+y\in \mathscr{F}(S)$,
  where $\mathscr{F}$ is either $\T$ or $\N$.

  Verification and refutation of algebraic ideal properties is
  important as well.
\begin{lemma}\label{P:nsalgideal}
  In every c-lattice,
  \begin{enumerate}
  \item $x\in \N(S)$ implies $x\para y\in \N(S)$,
  \item $x\in \T(S)$ implies $x\cdot y\in \T(S)$ and $y\cdot x \in \T(S)$.
  \end{enumerate}
\end{lemma}

\begin{lemma}\label{P:nsalgidealcounter}
  There are  $R\in \T(X)$ and $S\in \N(X)$ with$R\cdot S\not\in
  \N(X)$, $S\cdot R\not\in \N(X)$ and $R\para S\not\in \T(X)$.
\end{lemma}
\begin{proof}
  Let $R=\{(a,\emptyset)\}$ and $S=\{(a,\{a\})\}$. Then $R\cdot
  S=S\cdot R=R$ is not in $\N(X)$ and $R\para S=S$ is not in $\T(X)$.
\end{proof}

Define the relations on a c-lattice $S$ by
\begin{equation*}
  x \sqsubseteq_\tau y \Leftrightarrow \tau(x)\le \tau(y), \qquad x \sqsubseteq_\nu y \Leftrightarrow \nu(x)\le \nu(y). 
\end{equation*}

\begin{lemma}\label{P:nutauprecongprops}
  In every c-lattice,
  \begin{enumerate}
  \item the relations $\sqsubseteq_\tau$ and $\sqsubseteq_\nu$ are 
    partial orders,
\item $x \sqsubseteq_\tau y$ implies $x +z\sqsubseteq_\tau y+z$,
  $x\sqcap z\sqsubseteq_\tau y\sqcap z$, $x\para z\sqsubseteq_\tau 
  y\para z$ and $z\cdot x\sqsubseteq_\tau z\cdot y$,
\item $x \sqsubseteq_\nu y$ implies $x +z\sqsubseteq_\nu y+z$,
  $x\sqcap z\sqsubseteq_\nu y\sqcap z$ and $x\cdot z\sqsubseteq_\nu 
  y\cdot z$. 
  \end{enumerate}
\end{lemma}
The missing precongruence properties are justified by the following counterexamples.
\begin{lemma}\label{P:nutauprecongcounter}
  \begin{enumerate}
  \item There are $R,S,T\in \M(X)$ such that $\tau(R)\subseteq\tau(S)$
    and $\tau(R\cdot T)\not\subseteq\tau(S\cdot T)$. 
 \item There are $R,S,T\in \M(X)$ such that $\nu(R)\subseteq\nu(S)$
    and $\nu(R\para T)\not\subseteq\nu(S\para T)$. 
 \item There are $R,S,T\in \M(X)$ such that $\nu(R)\subseteq\nu(S)$
    and $\nu(T\cdot R)\not\subseteq\nu(T\cdot S)$. 
  \end{enumerate}
\end{lemma}
\begin{proof}
  \begin{enumerate}
  \item Let $R=\{(a,\{a\})\}$, $S=\emptyset$ and
    $T=\{(a,\emptyset)\}$. Then 
\begin{equation*}
\tau(R)=\tau(S)=\tau(S\cdot T)=\emptyset\subset T=
\tau(R\cdot T).
\end{equation*}
\item Let $R=\{(a,\emptyset)\}$, $S=\emptyset$ and
  $T=\{(a,\{a\})\}$. Then 
  \begin{equation*}
 \nu (R)=\nu(S)=\nu(S\para T)=\emptyset \subset \nu(R\para T)=T.
  \end{equation*}
\item Let $R=\{(a,\emptyset),(b,\{b\})\}$, $S=\{(b,\{b\})\}$ and
  $T=\{(a,\{a,b\})\}$. Then $\nu(R)=\nu(S)=S$, but $\nu(T\cdot
  R)=\{(a,\{b\})\}\supset \emptyset = \nu(T\cdot B)$.
\end{enumerate}
\end{proof}
Similarly, we can define the relations
\begin{equation*}
  x \sim_\tau y = x \sqsubseteq_\tau y \wedge y\sqsubseteq_\tau
  x,\qquad x \sim_\nu y = x\sqsubseteq_\nu y \wedge y \sqsubseteq_\nu x.
\end{equation*}
Obviously, therefore, $x\sim_\tau y \Leftrightarrow \tau(x)=\tau(y)$
and $x\sim_\nu y \Leftrightarrow \nu(x)=\nu(y)$. The following facts
then follow immediately from Lemma~\ref{P:nutauprecongprops} and~\ref{P:nutauprecongcounter}.
\begin{corollary}\label{P:nutaucongprops}
  In every c-lattice,
  \begin{enumerate}
  \item the relations $\sim_\tau$ and $\sim_\nu$ are equivalences,
\item $x \sim_\tau y$ implies $x +z\sim_\tau y+z$,
  $x\sqcap z\sim_\tau y\sqcap z$, $x\para z\sim_\tau 
  y\para z$ and $z\cdot x\sim_\tau z\cdot y$,
\item $x \sim_\nu y$ implies $x +z\sim_\nu y+z$,
  $x\sqcap z\sim_\nu y\sqcap z$ and $x\cdot z\sim_\nu 
  y\cdot z$. 
  \end{enumerate}
\end{corollary}
\begin{corollary}\label{P:nutaucongcounter}
  \begin{enumerate}
  \item There are $R,S,T\in \M(X)$ such that $\tau(R)=\tau(S)$
    and $\tau(R\cdot T)\neq\tau(S\cdot T)$. 
 \item There are $R,S,T\in \M(X)$ such that $\nu(R)=\nu(S)$
    and $\nu(R\para T)\neq\nu(S\para T)$. 
 \item There are $R,S,T\in \M(X)$ such that $\nu(R)=\nu(S)$
    and $\nu(T\cdot R)\neq\nu(T\cdot S)$. 
  \end{enumerate}
\end{corollary}

These results confirm our intuition about terminal and nonterminal
elements.  First, parallel composition with a nonterminal elements
yields a nonterminal element and sequential composition with a
terminal element yields a terminal element, whereas sequential
composition with a nonterminal element does not necessarily yield a
nonterminal element and parallel composition with a terminal element
does not necessarily yield a terminal element. Second, if two
multirelations agree on their terminal elements, then their parallel
compositions with a third element and their sequential composition
from the left with a third element also agree on their terminal
elements, whereas this need not be the case for sequential composition
from the right. If two multirelations agree on their nonterminal
elements, then their sequential compositions from the right by a third
element agree on their nonterminal elements, but this need not be the
case for sequential composition from the left or parallel composition.


\section{c-Quantales and Finite Iteration}\label{S:finiteiteration}

Iteration is best studied in a quantale setting where various
fixpoints exist. In fact, in our Isabelle formalisation, many of the
results in this and the following section are obtained in the weaker
settings of c-Kleene algebras and c-$\omega$-algebras~\cite{Struth15},
but quantales provide a unifying generalisation. In addition, the
least and greatest fixpoints corresponding to finite and infinite
iteration exists in quantales, whereas they need to be postulated in
the weaker algebras.

Let $(L,\le)$ be a complete lattice. We write $\sum X$ for the 
 supremum of the set $X\subseteq L$ and $\prod X$ for its infimum. In 
 particular, we write $x+y$ for the binary supremum and $x\sqcap y$
 for the binary infimum of $x,y\in L$. We write $U=\sum L$ for the 
 greatest and $0=\sum\emptyset$ for the least element of the lattice. 

 A \emph{proto-quantale} is a structure $(Q,\le,\cdot)$ such that 
 $(Q,\le)$ is a complete lattice and 
 \begin{equation*}
   x\le y\Rightarrow z\cdot x\le z\cdot y,\qquad (\sum_{i\in I} x_i)\cdot y=\sum_{i\in I} (x_i\cdot y). 
 \end{equation*}
 A proto-quantale is \emph{unital} if it has a unit of multiplication 
 $1$ which satisfies $1\cdot x = x$ and $x\cdot 1=x$. It is 
 \emph{commutative} if multiplication is: $x\cdot y=y\cdot x$ and 
 \emph{distributive} if the underlying lattice is distributive. 

 A \emph{quantale} is a proto-quantale with associative
 multiplication, $x\cdot (y\cdot z)=(x\cdot y)\cdot z$, and the
 following distributivity law holds:
 \begin{equation*}
   x\cdot (\sum_{i\in I}y_i) = \sum_{i\in I} (x\cdot y_i). 
 \end{equation*}

 A \emph{proto-bi-quantale} is a structure $(Q,\le,\cdot,\para
 ,\sid,\pid)$ such that $(Q,\le,\cdot,\sid)$ is a unital
 proto-quantale and $(Q,\le,\para ,\pid)$ a unital commutative
 quantale. Obviously, every pb-quantale is a proto-trioid. A
 \emph{c-quantale} is a proto-bi-quantale which is also a c-lattice.

 \begin{theorem}\label{P:cquantalesound}
   $(\M(X),\subseteq,\cdot,\para ,\sid,\pid)$ forms a boolean c-quantale. 
 \end{theorem}
 All axioms except for the infinite distributivity laws with respect
 to sequential and parallel composition have either been verified
 in Proposition~\ref{P:mrclattice}, or they follow directly from the underlying set
 structure.  Verification of these distributivity laws is
 straightforward.

 It follows from general fixpoint theory that all isotone functions on
 a quantale (as a continuous lattice) have least and greatest
 fixpoints. In addition, least fixpoints can be iterated from the least
 element of the quantale up to the first infinite ordinal whenever the
 underlying function is continuous, that is, it distributes with
 arbitrary suprema. Similarly, greatest fixpoints can be iterated from
 the greatest element of the quantale whenever the underlying function
 is co-continuous, which, however, is rarely the case.

 Furusawa and Struth~\cite{FurusawaS14} have studied the least
 fixpoints of the function
\begin{equation*}
  F_{RS}=\lambda X.\ S\cup R\cdot X
\end{equation*}
and its instance $F_R=F_{R\sid}=\lambda X.\ \sid\cup R\cdot
X$ in the concrete case of multirelations. Here we
study them abstractly in c-quantales. We write $x^\ast =\mu F_x$ and $x^\ast y =
\mu F_{xy}$.

Our first statement expresses $x^\ast$ in terms of its terminal and
nonterminal parts, at least in a special case.
\begin{lemma}\label{P:astsplit}
  In every c-quantale in which $x\cdot (y\cdot z)= (x\cdot 
  y)\cdot z$ holds for all elements,
  \begin{equation*}
    x^\ast = \nu(x)^\ast\cdot (\sid+\tau(x)). 
  \end{equation*}
\end{lemma}
\begin{proof}
  For $x^\ast \le \nu(x)^\ast (\sid+\tau(x))$ it suffices to show that
  that $\nu(x)^\ast\cdot (\sid+\tau(x))$ is a pre-fixpoint of $F_x$.
  In addition, 
  $\nu(x)^\ast (\sid+\tau(x))\le x^\ast$  follows from results
  from~\cite{FurusawaS14} , namely that $(x^\ast x^\ast)\le x^\ast$ (because
  $x^\ast + x\cdot x^\ast\le x^\ast$) and that $x^\ast \cdot x^\ast
  \le (x^\ast x^\ast)$ (by fixpoint fusion), whence $x^\ast \cdot
  x^\ast \le  x^\ast$. 
\end{proof}
Next we characterise the terminal and nonterminal parts of $x^\ast$.
\begin{lemma}\label{P:starnutau}
In every c-quantale,
  \begin{enumerate}
  \item     $\tau(x)^\ast = \sid+\tau(x)$,
  \item $\tau(x^\ast)= \nu(x^\ast)\cdot\tau(x)$ if $x\cdot (y\cdot
    z)=(x\cdot y)\cdot z$,
  \item $\nu(x)^\ast\le\nu(x^\ast)$,
\item $\tau(\nu(x)^\ast)=0$ and  $\nu(\nu(x)^\ast)=\nu(x)^\ast$,
\item $\nu(\tau(x)^\ast)=\sid$ and  $\tau(\tau(x)^\ast)=\tau(x)$. 
\end{enumerate}
\end{lemma}
Of course the lack of characterisation for $\nu(x^\ast)$ owes to that
for $\nu(x\cdot y)$.
\begin{lemma}\label{P:nustarcounter}
  There is a $R\in \M(X)$ such that $\nu(R^\ast)\not\subseteq\nu(R)^\ast$.  
\end{lemma}
\begin{proof}
Let $R=\{(a,\{b,c\}),(b,\emptyset),(c,\{d\})\}$. Then $(a,\{d\})\in 
R^\ast$ and $(a,\{d\})\in\nu(R^\ast)$, but $(a,\{d\})\not \in 
\nu(R)^\ast$. 
\end{proof}

Next we relate $x^\ast$ with notions of finite iteration.  For
multirelations, Peleg~\cite{Peleg87} has shown that $F_R$ and $F_{RS}$ are
not necessarily continuous, whereas in the \emph{externally image
  finite} case, where for each $(a,A)\in R$ the set $A$ has finite
cardinality,
\begin{equation*}
R^\ast=\mu F_R=F_R^\ast(\emptyset)=\bigcup_{i\in\mathbb{N}}F^i(\emptyset).
\end{equation*}
In addition, Furusawa and Struth~\cite{FurusawaS14} have defined
\begin{equation*}
  R^{(0)}=\emptyset,\qquad R^{(i+1)}=\sid\cup R\cdot R^{(i)},\qquad R^{(\ast)}=\bigcup_{i\in\mathbb{N}} R^{(i)}
\end{equation*}
and shown that, in the externally image finite case, $ R^\ast =
R^{(\ast)}$.  Finally,  Goldblatt~\cite{Goldblatt92}
has defined
\begin{equation*}
  R^{[0]} = \sid,\qquad  R^{[n+1]}= \sid\cup R\cdot R^{[n]},\qquad  R^{[\ast]} = \bigcup_{n\in\mathbb{N}} R^{[n]}. 
\end{equation*}

We now compare these iterations in the c-quantale setting. We call a
c-quantale \emph{externally image finite} if $x^\ast =
x^{(\ast)}$. First we prove a technical lemma.
\begin{lemma}\label{P:clatpowerprops}
In every proto-dioid,
  \begin{enumerate}
  \item $x^{(n)} \le x^{(n+1)}$,
 \item $x^{[n]} \le x^{[n+1]}$,
  \item   $(\sid+x)^n \le (\sid + x)^{n+1}$,
\item   $(\sid + x)^{n+1} = \sid+x \cdot (\sid+x)^n$.
  \end{enumerate}
  \end{lemma}

The next lemma shows that Goldblatt's iteration coincides with
Peleg's, and it gives a simpler characterisation of the former.
\begin{lemma}\label{P:iterpeleggoldblatt}
  In every c-quantale,
  \begin{enumerate}
\item $x^{(\ast)} = x^{[\ast]}$,
\item $  x^{[ \ast]} = \sum_{n\in\mathbb{N}} (\sid+x)^n$,
\item $\nu(x)^{(\ast)} = \sum_{n\in\mathbb{N}} (\sid+\nu(x))^n$.
  \end{enumerate}
\end{lemma}

Obviously, (3) follows from (1) and (2). Therefore, the
characterisation of $x^\ast$ can sometimes be simplified in the
externally image finite case.
\begin{corollary}\label{P:staritersplit}
  In every c-quantale in which $x\cdot (y\cdot z)= (x\cdot 
  y)\cdot z$ holds for all elements,
  \begin{equation*}
    x^{(\ast)} = (\sum_{n\in\mathbb{N}} (\sid +\nu(x))^n)\cdot (\sid +\tau(x)).
  \end{equation*}
\end{corollary}


\section{c-Quantales and Infinite Iteration}\label{S:divergence}

This section studies three additional notions of iteration for
multirelations: a unary strictly infinite iteration $R^\omega$, a
possibly infinite iteration $R^\infty$ and a binary possibly infinite
iteration $R^\omega S$. These arise as the greatest fixpoints of the
function
\begin{equation*}
 F_{RS}=\lambda X.\ S\cup R\cdot X
\end{equation*}
and its instances $F_R= F_{R\emptyset}=\lambda X.\ R\cdot X$ and
$F_{R\sid}=\lambda X.\ \sid\cup R\cdot X$. More precisely, $R^\omega
=\nu F_R$, $R^\infty=\nu F_{R\sid}$ and $R^\omega S=\nu F_{RS}$ and
these fixpoints exist due to the complete lattice structure of $\M(X)$
and the fact that the three functions under consideration are isotone
on that lattice.

As in the cases of the least fixpoints $\mu F_R$ and $\mu F_{RS}$, we
wish to relate the greatest fixpoints. We can use the following fusion
law for greatest fixpoints.
\begin{theorem}\label{P:nufusion}
  \begin{enumerate}
  \item Let $f$ and $g$ be isotone functions; let $h$ be a co-continuous function over a complete lattice. If単 $h\circ g \ge f\circ h$, then $h(\nu g)\ge \nu f$.
\item Let $f$, $g$ and $h$ be isotone functions over a complete lattice. If $f\circ h\ge h \circ g$, then $\nu f \ge h(\nu g)$. 
  \end{enumerate}
\end{theorem}
\begin{corollary}\label{P:nufusioncor}
  Let $R,S\in \M(X)$. Then
  $    R^\omega \cup\mu F_{RS}\subseteq \nu F_{RS}$.
\end{corollary}
The proof uses Theorem~\ref{P:nufusion}. In fact, we have also proved
a corresponding abstract statement in the context of proto-quantales
with Isabelle/HOL~\cite{Struth15}. However, the inequality is strict.
\begin{lemma}\label{P:nucounter}
  There are $R,S\in \M(X)$ such that $\nu F_{RS}\neq
  R^\omega \cup\mu F_{RS}$.
\end{lemma}
\begin{proof}
  Consider the multirelations $R=\{(a,\{b,c\}),(b,\{a\})\}$ and $S =
  \{(c,\{a\})\}$ over the set $X=\{a,b,c\}$.  Then $R^\omega
  =\emptyset$ can be checked easily. 

  Moreover, $R^\ast
  S=\{(a,\{a\}),(a,\{b,c\}),(b,\{a\}).(b,\{b,c\}),(c,\{a\})\}$. This
  can be checked by verifying  $S\cup R\cdot (R^\ast
  S)=(R^\ast S)$. It follows that
\begin{equation*}
R\cdot (R^\ast
  S)=\{(a,\{a\}),(a,\{b,c\}),(b,\{a\}),(b,\{b,c\})\}.
\end{equation*}

Finally, $R^\omega S=\{(a,S),(b,S),(c,a)\}$ for all $S\subseteq X$,
which can be checked by verifying $S\cup R\cdot (R^\omega S)=(R^\omega
S)$. In particular, $R\cdot (R^\omega S)=\{(a,S),(b,S)\}$.
\end{proof}

This counterexample is not related to the absence of associativity,
but to the lack of left distributivity. There is therefore no hope
that the situation can be resurrected for tests and modalities, as in
the case of the star~\cite{FurusawaS14},

As a consequence, the greatest fixpoint $R^\omega S$ cannot be reduced
to a formula involving the greatest fixpoint $R^\omega$. At the level
of c-quantales we obtain the fixpoint axioms
\begin{equation*}
  x^\omega y \le y+x\cdot (x^\omega y),\qquad  z \le y + x\cdot z \Rightarrow z\le x^\omega y.
\end{equation*}
Since $F_R=F_{R\emptyset}$ by definition, this yields $x^\omega =
x^\omega 0$ and the following unary $\omega$-unfold and
$\omega$-coinduction axioms as special cases:
\begin{equation*}
  x^\omega \le x\cdot x^\omega, \qquad  y \le x\cdot y \Rightarrow y\le x^\omega.
\end{equation*}
These can be used for deriving the following properties.
\begin{lemma}\label{P:omegaprops}
For every c-quantale,
  \begin{enumerate}
  \item $x^\omega = x\cdot x^\omega$,
  \item $x\le y \Rightarrow x^\omega \le y^\omega$.
  \end{enumerate}
\end{lemma}
\begin{lemma}\label{P:omegaconsts}
For every c-quantale,
  \begin{enumerate}
  \item $0^\omega = 0$,
\item $\pid^\omega = \pid$,
  \item $\sid^\omega =\pidc^\omega = U^\omega = U$,
   \end{enumerate}
\end{lemma}
The fact that $\sid^\omega$ and $\pidc^\omega$ are $U$, instead of
$\pidc$, is not entirely satisfactory: one would not assume that the
iteration of an element $\N(S)$ leads outside of this set.

The characterisation of terminal parts of elements is still
satisfactory, whereas for nonterminal elements, the lack of
characterisation for products makes the situation less pleasant.

\begin{lemma}\label{P:tauomega}
In every c-quantale,
  \begin{enumerate}
  \item $\tau(x)\le \tau(x^\omega)$,
  \item $\tau(x)^\omega = \tau(x)$,
\item $\tau(x)^\omega \le \tau(x^\omega)$.
  \end{enumerate}
\end{lemma}
\begin{lemma}\label{P:omegasplit}
In every c-quantale,  $\nu(x)^\omega + \nu(x)^\ast\cdot\tau(x)\le x^\omega$.
\end{lemma}
The converse implication is ruled out by the counterexample in Lemma~\ref{P:nucounter}.

Next we briefly compare $F_{R\sid}$ with $F_{RS}$, that is, $R^\infty$
with $R^\omega S$.

\begin{lemma}\label{P:inftycounter}
  There are $R,S\in \M(X)$ such that $(R^\omega S)\neq
  R^\infty\cdot S$.
\end{lemma}
\begin{proof}
  In the above counterexample, $R^\omega S=\{(a,S),(b,S),(c,a)\}$ for
  all $S\subseteq X$. However, $R^\infty=R^\ast=
  \{(a,\{a\}),(a,\{b,c\}),(b,\{b\}),(b,\{a\}),(b,\{b,c\})\}$ and
  therefore $R^\infty\cdot S= R^\ast\cdot S=\emptyset$.
\end{proof}
Here we cannot even obtain $R^\infty\cdot S\subseteq (R^\omega S)$ by
fixpoint fusion, since this requires co-continuity of $\lambda x.\
x\cdot S$, that is, $(\bigcap_{i\in I}R_i)\cdot S=\bigcap_{i\in
  I}(R_i\cdot S)$, which does not hold. Hence this time the failure of
equality owes to the lack of associativity and co-continuity. Note
that $(R^\omega S)$ need not be equal to $ R^\infty\cdot S$ for binary
relations for similar reasons. In this case, $x^\infty = x^\omega
\sid$ by definition of $F_{R\sid}$ and $F_{RS}$, which yields the
unary laws
\begin{equation*}
  x^\infty\le \sid + x\cdot x^\infty,\quad y\le \sid + x\cdot y\Rightarrow y\le x^\infty.
\end{equation*}

Finally, following~\cite{DesharnaisMS11}, we study a notion of
greatest fixpoint on c-quantales which models the set of all states in
$x$ from which either $\emptyset$ is reachable or infinite $x$-chains
start. Obviously, the set of domain elements is the complete
distributive lattice or boolean subalgebra of the sequential
subidentities. The function $\lambda p.\ q+\langle x\rangle p$ is
isotone and thus has a least (binary) fixpoint $\langle x\rangle ^\ast
q$ which, by the results of~\cite{DesharnaisMS11}, is equal to
$\langle x^\ast\rangle q$. In the context of c-monoids, $\langle
x\rangle p = d(x\cdot p)$.

For the same reasons, the function $\lambda p.\ \langle x\rangle p$ has a greatest
fixpoint which we denote $\nabla x$. It satisfies the unfold and
coinduction axiom
\begin{equation*}
  \nabla x \le \langle x\rangle \nabla x,\qquad p \le \langle x\rangle p \Rightarrow p \le \nabla x.
\end{equation*}
We can use fixpoint fusion again and try to derive the rule
\begin{equation*}
  p \le \langle x\rangle p + q \Rightarrow p \le \nabla x + \langle x^\ast\rangle q.
\end{equation*}
We must instantiate $f = \lambda y. \langle x\rangle y +q$, $g= \lambda y.\ \langle x\rangle y$ and $h= \lambda y.\ \langle x^\ast\rangle q$. Then
\begin{equation*}
  h(g(y)) =\langle x\rangle y + \langle x^\ast\rangle q
=\langle x\rangle y+ q +\langle x\rangle\langle x^\ast\rangle q
\le \langle x\rangle (y+ \langle x^\ast\rangle q) +q
= f(h(y)).
\end{equation*}
Then greatest fixpoint fusion yields once more
\begin{equation*}
  \nabla x + \langle x^\ast\rangle q \le \nu f,
\end{equation*}
but the above counterexample excludes equality without $\langle
x\rangle (p+q)=\langle x\rangle p+\langle x\rangle q$.

In general, $\nabla x$ is the largest subidentity $p$ which satisfies
$p=d(x\cdot p)$. Hence $p$ models the largest set of states from which
executing $x$ either leads to $\emptyset$, or it leads back into $p$ in
the following sense. For each element in $p$ there exists a set $A$
which is reachable via $X$, and all elements of $A$ are in $p$. This
means that from all states in $p$ indeed either infinite executions with $x$
are possible or $\emptyset$ is reachable.

In this case, at least an explicit definition of $\nu(x)^\omega$ is
possible.
\begin{proposition}\label{P:nuomegadef}
  Let $S$ be a c-quantale. If $x\cdot (d(y)\cdot z)= (x \cdot
  d(y))\cdot z$ for $x,y,z\in S$, then
\begin{equation*}
  \nu(x)^\omega = \nabla(\nu(x))\cdot U. 
\end{equation*}
\end{proposition}

As in Lemma~\ref{P:omegaconsts}, the infinite iteration of nonterminal
elements contains terminal parts, which seems undesirable, but
unavoidable in this context.

\begin{corollary}\label{P:nuomegaapprox}
  Let $S$ be a c-quantale. If $x\cdot (d(y)\cdot z)= (x \cdot 
  d(y))\cdot z$ for all $x,y,z\in S$, then
  \begin{equation*}
    \nabla(\nu(x))\cdot U +\nu(x)^\ast\cdot\tau(x)\le x^\omega.
  \end{equation*}
\end{corollary}

At least for nonterminal multirelations, the situation is the same. In
the general case, terminal parts contribute to infinite iteration,
too. The absence of associativity and distributivity makes the
situation more complicated. Separation of infinite iterations into
terminating and nonterminating parts yields at least an
underapproximation. Using $\nabla$ also yields sharper properties for
nonterminal elements.
\begin{lemma}\label{P:nuomeganabla}
  Let $S$ be a c-quantale. If $x\cdot (d(y)\cdot z)= (x \cdot 
  d(y))\cdot z$ for all $x,y,z\in S$, then
  \begin{enumerate}
  \item $\nu(\nu(x)^\omega) = \nabla (\nu(x))\cdot\pidc$,
  \item $\tau(\nu(x)^\omega) = \nabla (\nu(x))\cdot\pid$.
  \end{enumerate}
\end{lemma}

Lemma~\ref{P:nuomeganabla}(2) clearly shows that the infinite
iteration of a nonterminal multirelation has a terminal and a
nonterminal part.  Hence there is no direct relationship between
strictly infinite iteration and terminal or nonterminal elements.
This is in contrast to the language case, where $\mathsf{inf}$ models
infinite or divergent behaviour. The latter, according to
Lemma~\ref{P:nuomegadef}, is captured by $\nabla$ and or $U$ in the
multirelational model, which is similar to the relational case. As
suggested by alternating automata, terminal parts of a multirelation
rather model success or failure states, or winning or loosing states
in a game based scenario, but not nontermination.

The final part of this section sets up the correspondence between
infinite iteration and notions of deflationarity and wellfoundnedness,
as they have been studied in the relational model by~\cite{Struth12}.  We
call an element $x$ \emph{$\omega$-trivial} if $x^\omega = 0$,
\emph{deflationary} if $\forall y.\ (y\le x\cdot y\Rightarrow y=0)$
and \emph{wellfounded} if $\forall y.\ (d(y) \le d(x\cdot
y)\Rightarrow d(y)=0)$.
\begin{proposition}\label{P:wfequiv}
Let $x$ be an element of a c-quantale. The following statements are equivalent:
\begin{enumerate}
\item $x$ is wellfounded;
\item $x$ is deflationary;
\item $x$ is $\omega$-trivial.
\end{enumerate}
\end{proposition}
In this respect, the behaviour of relations and multirelations is the
same.



\section{Counterexamples}\label{S:counterexamples}

This section collects some counterexamples for multirelations. The
second set, in particular, rules out variants of interchange laws, as
they arise for instance in the context of monoidal categories, in
shuffle languages or for partially ordered
multisets~\cite{BloomEsik,Gischer88,Kelly}. More abstractly, such laws
have been considered in concurrent Kleene algebras~\cite{HoareMSW11}.

\begin{lemma}\label{P:mrcounter1}
  There exist multirelations $R,S,T\in \M(X)$ such that 
  \begin{enumerate}
  \item $R\para R\not\subseteq R$,
  \item $R\not\subseteq R\para S$,
\item $R\para S\cap R\para T\not\subseteq R\para (S\cap T)$,
\item $(R\cdot S)\cdot T\not\subseteq R\cdot (S\cdot T)$,
\item $(R\cdot T)\para (S\cdot T)\not\subseteq (R\para S)\cdot T$,
  \item $(R\cdot S)\para  (R\cdot T)  \not\subseteq R\cdot (S\para T)$, even for $R\para R\subseteq R$, $S\para S\subseteq S$ and $T\para T\subseteq T$. 
  \end{enumerate}
\end{lemma}
\begin{proof}
  \begin{enumerate}
  \item Let $R=\{(a,\{a\}),(b,\{b\})\}$. Then $R\para R = \{(a,\{a\}),(b,\{b\}),(a,\{a,b\})\}\not\subseteq R$. 
  \item Let $R=\{(a,\{a\})\}$ and $S=\{(a,\{a,b\})\}$. Then 
    $R\not\subseteq S = R\para S$. 
  \item Let $R=\{(a,\{b,c\})\}$, $S=\{(a,\{b\})\}$ and 
    $T=\{(a,\{c\})\}$. Then 
    \begin{equation*}
      R\para S\cap R\para T=R\para S=R\para T=R\supset \emptyset = S\cap T = R\para (S\cap T). 
    \end{equation*}
\item A counterexample has been given in~\cite{FurusawaS14}. 
\item A counterexample  has again been given in~\cite{FurusawaS14}. 
  \item Let $R=\{(a,\{a\}),(a,\emptyset)\}$, $S=\{(a,\{a\})\}$ and 
    $T=\emptyset$. Then 
    \begin{equation*}
      R\cdot (S\para T) = \{(a,\emptyset)\} \subset \{(a,\{a\}),(a,\emptyset)\} = (R\cdot S) \para (R\cdot T). 
    \end{equation*}
    It is straightforward to check that $R\para R\subseteq R$,
    $S\para S\subseteq S$ and $T\para T\subseteq T$. 
  \end{enumerate}
\end{proof}

Our next counterexamples explains the difference between algebras of
multirelations and concurrent Kleene algebras. The latter are based on
a full sequential dioid and a commutative one, with shared units of
sequential and concurrent composition. The sequentiality-concurrency
interaction is captured by the interchange law
\begin{equation*}
  (w\para x)\cdot (y\para z)\le (w\cdot y)\para (x\cdot z). 
\end{equation*}
From this law, the small interchange laws $ (x\para y)\cdot z\le
x\para (y\cdot z)$, $x\cdot (y\para z)\le (x\cdot y)\para z$ and
$x\cdot y\le x\para y$ are derivable in the presence of a shared unit
for sequential and concurrent composition.  These laws hold, in
particular, in certain pomset languages and in word languages under
the regular operations and with concurrent composition interpreted as
shuffle. However, the next lemma refutes all variants of interchange
between sequential and concurrent composition.  We write $R\asymp S$
if $R\not\subseteq S$ and $R\not\supseteq S$.

\begin{lemma}\label{P:interchangecounter}
  There exist $R,S,T,U\in \M(X)$ such that 
  \begin{enumerate}
  \item $(R\para S)\cdot (T\para U)\asymp (R\cdot T)\para (S\cdot U)$,
\item $(R\para S)\cdot T\asymp R\para (S\cdot T)$,
\item $R\cdot (S\para T)\asymp (R\cdot S)\para T$,
\item $R\cdot S\asymp R\para S$. 
  \end{enumerate}
\end{lemma}
\begin{proof}
  \begin{enumerate}
  \item Let $R=\{(a,\{a\}),(b,\{a,b\})\}$ and $S=\{(a,\{a\}),(b,\{a\})\}$. Then 
    \begin{equation*}
      (R\para S)\cdot (R\para S) = \{(a,\{a,b\}),(b,\{a,b\})\} \asymp \{(a,\{a\})\} =(R\cdot R) \para  (S\cdot S). 
    \end{equation*}
\item Let $R= \{(a,\{a,b\})\} = S$ and 
  $T=\{(a,\{a\}),(b,\{a\})\}$. Then 
  \begin{equation*}
    (R\para R) \cdot T = R \cdot T = \{(a,\{a\})\} \asymp R = R\para (R\cdot T). 
  \end{equation*}
\item Let $R = \{(a,\{a,b\})\}$, $S=\{(a,\{a\}),(b,\{a\})\}$ and $T= \{(a,\{a\}),(b,\{a,b\})\}$. Then 
  \begin{equation*}
    R\cdot (S\para T) = R \cdot T = R \asymp \{(a,\{a\})\} = (R\cdot S)\para T 
  \end{equation*}
\item Let $R = \{(a,\{a,b\})\}$ and $S=\{(a,\{a\}),(b,\{a\})\}$. Then 
  \begin{equation*}
R\cdot S = \{(a,\{a\})\} \asymp \{(a,\{a,b\})\} = R \para S. 
\end{equation*}
  \end{enumerate}
\end{proof}


\section{Remarks on Up-Closed Multirelations}\label{S:upclosed}


This section briefly discusses the relationship between general
multirelations \`a la Peleg and the up-closed multirelations studied
by Parikh and others. As explained in Section~\ref{S:seqsubids}, a
multirelation is up-closed if $(a,A)\in R$ and $A\subseteq B$ imply
$(a,B)\in R$. This is the case if and only if $R=R\para U$.

For this special case, Parikh has defined sequential composition as
\begin{equation*}
  (a,A)\in R;S \Leftrightarrow \exists B.\ (a,B)\in R \wedge \forall b\in B.\ (b,A)\in S.
\end{equation*}
This corresponds to Peleg's definition with $f$ instantiated to
$\lambda x.A$.  Peleg's definition, however, does not preserve
up-closure.  Let $R'=\{(a,\emptyset)\}$, $R=R'\para U$ and
$S=\emptyset$. Then $R,S\in\U(X)$, whereas $R\cdot
S=R'\not\in\U(X)$. Hence Peleg's definition does not simply specialise
to Parikh's when multirelations are up-closed.  The next lemma
captures the precise relationship.
\begin{lemma}\label{P:parikhpeleg}
  Let $R,S\in \M(X)$. Then $R;(S \para U) = (R\cdot S)\para U$.
\end{lemma}
\begin{equation*}
  \def\labelstyle{\normalsize}
\xymatrix @R=2pc @C=5pc{
 \M(X)\times \M(X)\ar[r]^{1_{\M(X)}\times (\para U)}\ar[d]_{\cdot \ } &
   \M(X)\times \U(X)\ar[d]^{ \ ;}\\
\M(X)\ar[r]_{(\para U)} & \U(X) 
}
\end{equation*}
Consequently, Parikh's sequential composition preserves up-closure,
\begin{equation*}
  (R;(S\para U))\para U = (R \cdot S)\para U\para U = (R \cdot S)\para
  U = R;(S\para U),
\end{equation*}
and join.  Moreover, Parikh composition of up-closed multirelations
can be simulated by Peleg's as
$  (R\para U) ; (S \para U) = ((R\para U)\cdot (S\para U))\para U$.

Parikh composition of up-closed multirelations is associative.  Hence
it might be possible to derive this property in c-lattices, but it
does not seem straightforward, and specific associativity laws for
Peleg's composition in the presence of up-closed elements might be
needed. A deeper investigation is left for future work.

The case of parallel composition in the up-closed case is very simple
in comparison. Rewitzky and Brink~\cite{RewitzkyB06} have studied
parallel composition of up-closed multirelations under the name
\emph{power union} and shown that it simply yields a contrived
definition of meet. We can reproduce this result in the setting of
c-lattices.
\begin{lemma}\label{P:upclosedpar}
  In every c-lattice, 
\begin{enumerate}
\item$x\para U\sqcap y\para U = (x\para U)\para
  (y\para U)$,
  \item $(x\para U) \para (x\para U) = x\para U$,
  \item $(x\para y)\cdot (z\para U) = (x\cdot (z\para U))\para (y\cdot 
    (z\para U))$.
\end{enumerate}
\end{lemma}
\begin{corollary}\label{P:upclosedparcor}
  If $R,S\in \M(X)$ are up-closed, then $R||S=R\cap S$. 
\end{corollary}
By this result, up-closed multirelations are also meet closed, and
therefore closed under the usual operations, which is well known.

Our results for c-lattices and c-quantales hold automatically in the
up-closed case, interpreting parallel composition as intersection
and translating Peleg's sequential composition into Parikh's. In the
up-closed case, in particular, the interaction of sequential and
parallel composition becomes a sub-distributivity law for sequential
composition over meet, which follows directly from the greatest lower
bound properties of meet and isotonicity of sequential composition.

However there are differences as well. First, $\pid\para U =U$, hence
the up-closure of the unit of parallel composition is $U$, which is
consistent with $U$ being the unit of meet. More generally, up-closure
turns parallel subidentities into vectors and therefore
$\T(X)=\V(X)$ in that context. Second, the definitions of
subidentities, of domain as $d(R)=(R\cdot \pid)\para \sid$, and of the
corresponding box and diamond modalities, do not carry over directly
to the up-closed case. In particular, the definitions of the
sequential unit and of sequential subidentities are now based on the
$\in$-relation instead of the function $\lambda x.\{x\}$. We leave a
reconstruction of the subalgebra of up-closed elements in the context
of c-lattices for future work as well. In particular, Parihk's game
logic can be based on the domain and antidomain axioms for pre-dioids
given in~\cite{DesharnaisStruth08} and linked with concepts from
previous sections..

A duality between up-closed multirelations and sets of isotone
predicate transformers has already been noticed by Parikh~\cite{Parikh85}.
By this isomorphism, sequential composition of up-closed
multirelations is associated with function composition of monotone
predicate transformers. Obviously, a similar isomorphism between
Peleg's multirelations and a class of predicate transformers cannot
exist since sequential composition of multirelations is not
associative. Otherwise, a non-associative operation would have to be
defined on predicate transformers, which may not lead to a natural
concept.


\section{Conclusion}\label{S:conclusion}

We have investigated the structure and algebra of multirelations,
which model alternation or dual nondeterminism in a relational
setting, and which form the semantics of Peleg's concurrent dynamic
logic, extending a previous algebraic approach to concurrent dynamic
algebra. Apart from the derivation of a considerable number of
algebraic properties which arise from the union, intersection,
sequential and parallel composition, and finite and infinite iteration
of multirelations, we have also studied the structure of various
subalgebras and the relationships between them, as illustrated in the
diagrams of Figure~\ref{F:extendeddiagram}. In particular we found
that a domain operation, which is important for these investigations,
can be defined explicitly in this setting.

The operations on multirelations are rather complex; their
interactions are intricate.  Sequential composition, for instance,
requires a higher-order definition, its manipulation often depends on
the axiom of choice.  Algebraic axioms similar to Tarski's relation
algebra are therefore desirable to hide this complexity.  To address
this we have developed algebraic axiom systems ranging from c-monoids
to c-quantales and carried out most of our work at that level.  At the
moment, these axiom systems are less compact than those for up-closed
multirelations or even relation algebra. It seems that much of the
power of relation algebra comes from the operation of conversion, to
which we do not know a multirelational counterpart. There is certainly
scope for completing, revising and perhaps simplifying our axioms.

Hence, from a mathematical point of view, more concise and
comprehensive axiom systems seem desirable and questions such as
representability, axiomatisability, and decidability at least of
fragments are interesting---the class of representable relation
algebras, which are isomorphic to algebras of binary relations, is not
finitely axiomatisable~\cite{Monk64} and a similar result might be
expected here. Other directions for research include the investigation
of the up-closed case in relationship with c-lattices and c-quantales,
the study of other classes of multirelations in which sequential
composition is associative, the algebraic reconstruction of Parikh's
game logic and the association of multirelations with predicate
transformer algebras. Beyond that, a reference formalisation of
multirelational algebras in Isabelle, including their modal variants
such as dynamic and game logics, has been developed in parallel to
this article~\cite{Struth15}. Its use in the formal analysis of games
and the verification of computing systems with dual nondeterminism are
next steps towards taming multirelations.

\paragraph*{Acknowledgements.}
  The authors are grateful to Alasdair Armstrong for his help with
  intricate Isabelle proofs. They would like to thank Yasuo Kawahara,
  Koki Nishizawa, Toshiori Takai and Norihiro Tsumagari for helpful
  discussions. The second author is grateful to the Department of
  Mathematics and Computer Science at Kagoshima University and the
  Department of Computer Science at the University of Tokyo for their
  hospitality while working on this article.

\bibliographystyle{plain}
\bibliography{cda}

\end{document}